\documentclass[11pt]{article}

\usepackage[a4paper,left=1in,right=1in,top=1in,bottom=1in]{geometry}

\usepackage[authoryear]{natbib}
\usepackage{graphicx} 
\usepackage{amsmath,amsfonts,amssymb}
\usepackage{amsthm}
\usepackage{xcolor}
\usepackage{empheq}
\usepackage{hhline}
\usepackage{float}
\usepackage{subcaption}
\usepackage{multirow}
\usepackage{booktabs}
\usepackage{placeins}

\usepackage{hyperref}
\definecolor{mydarkblue}{RGB}{72,61,139}
\hypersetup{
 colorlinks=true,breaklinks=true,
 urlcolor= mydarkblue,linkcolor=mydarkblue,citecolor=mydarkblue,
 pdfauthor={Daluiso, R., Folgar-Came\'an, H., Pallavicini, A., V\'azquez, C.}
}

\newtheorem{theorem}{Theorem}[section]
\newtheorem{lemma}[theorem]{Lemma}
\newtheorem{corollary}[theorem]{Corollary}
\theoremstyle{definition}
\newtheorem{definition}[theorem]{Definition}

\theoremstyle{remark}
\newtheorem{remark}[theorem]{Remark}
\newtheorem{hypotheses}{Hypotheses}

\title{Rough volatility dynamics in commodity markets\footnote{The authors report no potential competing interests. The opinions expressed in this document are solely those of the authors and do not represent in any way those of their present and past employers.}}
 
\author{
R.~Daluiso\thanks{Intesa Sanpaolo, Financial Engineering. Address: largo Mattioli 3, Milano 20121, Italy. Email address: \texttt{roberto.daluiso@intesasanpaolo.com}.},
H.~Folgar-Came\'an\thanks{Department of Mathematics and CITIC, Universidade da Coruña. Address: Campus Elviña s/n, 15071-A Coruña, Spain. Email address: \texttt{hector.folgar.camean@udc.es}.},
A.~Pallavicini\thanks{Intesa Sanpaolo, Financial Engineering. Address: largo Mattioli 3, Milano 20121, Italy. Email address: \texttt{andrea.pallavicini@intesasanpaolo.com}.},
C.~V\'azquez\thanks{Department of Mathematics and CITIC, Universidade da Coruña. Address: Campus Elviña s/n, 15071-A Coruña, Spain. Email address: \texttt{carlosv@udc.es}.}}

\begin{document}

\maketitle

\begin{abstract}
In this paper, we develop a general rough volatility model for commodities that provides an automatic calibration of the initial term structure of the futures prices and an appropriate treatment of the Samuelson effect. After the theoretical analysis of this general model, we focus on the rBergomi and rHeston models and their calibration to market data of vanilla futures options on WTI Crude Oil. Finally, numerical results illustrate the performance of the proposed rough volatility models for commodities pricing.
\end{abstract}

\bigskip
 
\noindent {\bf AMS classification codes:} 65C05, 91G20, 91G60.

\newpage
{\small \tableofcontents}
\vfill
\newpage

\pagestyle{myheadings} \markboth{}{{\footnotesize  Daluiso, Folgar-Came\'an, Pallavicini, V\'azquez. Rough volatility dynamics in commodity markets.}}

\section{Introduction}

Historically, the Black-Scholes model \citep{blackscholespaper} has been used extensively to price option contracts. It provides a closed-form formula for pricing European options and laid the theoretical foundations for the modern theory of option pricing.

Among its various assumptions, the Black-Scholes model assumes that the volatility of the underlying asset price is constant; however, this assumption is inconsistent with empirical market data. Over the last half century, new models have been proposed to implement non-constant volatility and reproduce the market dynamics more accurately, such as stochastic, local, or local-stochastic volatility models.

\citet{gatheralpaper} analysed  the log-volatility increments for several equity indices and found that, for reasonable timescales of practical interest, these increments appeared normally distributed and were $H$-self-similar for some constant $H>0$. Moreover, the estimated Hurst exponent $H$ was found to be less than $0.5$, indicating rough behaviour in the log-volatility dynamics. Motivated by these results, the authors introduced the rough fractional stochastic volatility model (rFSV), in which the volatility is driven by a fractional Brownian motion (fBm) with a small Hurst parameter $H \sim 0.1$.

Following \citet{gatheralpaper}, rough volatility models (defined as models in which the volatility is driven by a continuous process with Hölder coefficient lower than $0.5$) have received considerable attention from both academics and practitioners, due to their ability to reproduce implied volatility surfaces more accurately. For instance, the ATM volatility skew of equity indices is empirically observed to decay exponentially with time to maturity, a stylised fact known as the Fukasawa effect. In classical stochastic volatility models the ATM volatility skew behaves constant for short times to maturity (see for instance \citet{bergomibook}), whereas \citet{fukasawapaper} showed formally that the ATM volatility skew generated by rough volatility models decays exponentially, in line with the empirical evidence. We refer the reader to \citet{nunnomishurapaper} for a more comprehensive survey of the historical development of rough volatility modelling.

Rough volatility models have been successfully applied in several markets, including equities \citep{gatheralpaper} or cryptocurrencies \citep{takaishipaper}. However, other financial markets remain relatively unexplored, such as the commodity markets. Our review of the existing literature suggests that the applicability of rough volatility in commodity markets has only been investigated in the work of \citet{alfeuspaper}, who estimated the Hurst parameter for several commodities and found that $H$ was consistently smaller than 0.2 across all commodity markets considered.

The objective of this paper is to explore the applicability of rough volatility models in the context of commodities and to price futures options within this market, building on the work of \citet{pallavicinipaper}.

Our motivation is two-fold: on the one hand, the price time-series in different commodities exhibit rough behaviour, as shown in \citet{alfeuspaper}. On the other hand, we also seek a model that allows for a suitable joint calibration of multiple contracts with different maturities, where the involved parameters correctly explain the market dynamics. In view of these motivations, we develop a general rough volatility model for commodities that provides an automatic calibration of the initial term structure of futures prices and an appropriate treatment of the Samuelson effect (see Section \ref{sec:samuelson}). After the theoretical analysis of this general model, we focus on the rBergomi and rHeston models and their calibration to market data of vanilla futures options on WTI Crude Oil.

The rest of the paper is organised as follows. In Section \ref{sec:market} we outline the main characteristics of commodity markets and introduce the mathematical concept of fictitious spot price, following the exposition in \citet{pallavicinipaper}. In Section \ref{sec:model}, we summarise the main properties of the fractional Brownian motion and introduce the proposed general rough volatility model, providing a mathematical justification and particularising for the rBergomi and rHeston models. In Section \ref{sec:calibration}, we discuss the simulation and calibration schemes for the models considered. In Section \ref{sec:results}, we present the numerical results for some real cases examples. Finally, in Section \ref{sec:conclusion}, we summarise and conclude.

\section{Commodity markets}\label{sec:market}

Commodity markets can be defined as financial markets in which raw materials and other physical goods are traded, such as crops, gold, and crude oil. The main characteristic of commodities is that they are physical objects that need to be stored and transported, involving logistical challenges when trading them. Derivative contracts allow traders to participate in the market without dealing with any of these issues.  Due to this, futures contracts are the most liquid commodity contracts, followed by options on futures \citep{pallavicinipaper}.  In the following, we recall the notation introduced in \citet{pallavicinipaper}, which we will use extensively.

Due to the physical nature of commodities, futures contracts typically stipulate a period during which the seller can deliver the goods at any time. We denote this period by $[T_0^d, T_1^d]$, where $T_0^d$ and $T_1^d$ are the start and final delivery dates, respectively.

The seller must notify the buyer a few days prior to delivery. Since the delivery may occur at any time within $[T_0^d, T_1^d]$, the notification can also occur within a specified period $[T_0^n, T_1^n]$, where $T_0^n$ and $T_1^n$ are the first and last notification dates, respectively. The only restriction is that the notification must occur before the delivery.

Futures option contracts can be traded over a period $[T_0^o, T_1^o]$, which lies before both $T_1^f$ and $T_0^n$; that is, prior to the end of the underlying futures trading period and the beginning of the notification period, respectively. Here, $[T_0^f, T_1^f]$ denotes the period during which the underlying futures contract can be traded.

We denote by $F_t(T)$ the price at time $t$ of a futures contract with maturity date $T$. We omit $T$ in the notation $F_t(T)$ when unambiguous.

In equity markets and under standard assumptions on absence of arbitrage, the future price converges to the underlying spot price at delivery, and an equivalent measure exists such that the futures price process is a martingale. This is not true in commodity markets due to the delivery procedure: in general, commodities spot prices may not coincide with futures prices at maturity, so the arguments used for equity markets cannot be directly transferred to commodities.

To address this issue, as previously done in the literature \citep{pallavicinipaper}, we assume the existence of a single time-continuous process $S_t$, which we call the \textbf{fictitious spot price}, representing the price of the rolling futures, so that its value coincides with futures contract prices on the $T^{\text{last}}$ date of each contract, that is:
\begin{equation}\label{def_fictitious_spot_price}
	F_t = \mathbb{E}_t\left[ S_{T^{\text{last}}} \right], \quad t\leq T^{\text{last}},
\end{equation}
where $T^{\text{last}} = \min\{T_0^n, T_1^f\}$ and $\mathbb{E}_t$ denotes the conditional expectation. From now on, we will refer to $T^{\text{last}}$ as the maturity date of the futures contract, unless stated otherwise. This choice is natural from a trading perspective, as $T^{\text{last}}$ is the last date on which trading the futures contract for financial exposure is reasonable.

\section{Rough volatility for commodities}\label{sec:model}
\subsection{Fractional Brownian motion}
Fractional Brownian motion (fBM for short) was first introduced theoretically in \citet{kolmogorovpaper} and later studied for its statistical properties in \citet{mandelbrotpaper}. Most of the literature on fractional volatility consists of processes driven by fBm, and we will follow the same approach in this study; therefore, it is advisable to introduce its main properties.

\begin{definition}[Fractional Brownian motion]
	A Gaussian stochastic process $\{W_t^H, \, t\geq0\}$ is said to be a \textbf{fractional Brownian motion} with \textbf{Hurst parameter} $H\in(0,1)$ if
	\begin{itemize}
		\item $\mathbb{E}[W_t^H] = 0$ for every $t\geq0$ and
		\item $\mathbb{E}[W_t^H W_s^H] = \frac{1}{2}\left( t^{2H} + s^{2H} - |t-s|^{2H} \right)$ for every $t,s\geq0$.
	\end{itemize}
\end{definition}
The standard Brownian motion is recovered when $H=0.5$. Hence, we can consider the fBm as a generalisation of the standard Brownian motion.

The Hurst parameter dictates the regularity of the paths. Indeed, there exists a modification of the process such that its trajectories are $\gamma$-Hölder continuous almost surely for any $\gamma \in (0,H)$, see for instance \citet{mishurabook}. Thus, when $H<0.5$ the trajectories will be rougher than those of the standard Brownian motion, and they will be smoother when $H>0.5$. Furthermore, we also remark the following properties of the fBm:
\begin{itemize}
    \item It is neither a semimartingale nor a Markov process for $H\neq 0.5$ (see \citet{rogerspaper}).
    \item It is $H$-self-similar (see \citet{mishurabook}).
    \item When $H > 0.5$, the increments of the fBm exhibit long memory, also known as long-term dependence (see \citet{mishurabook}).
\end{itemize}

A more comprehensive discussion on the fractional Brownian motion and its properties can be found in \citet{mishurabook}. We also refer the reader to \citet{nunnomishurapaper} for a more detailed review on the role of rough volatility and fBm in quantitative finance.

\subsection{General rough volatility model for commodities}

Building on the work of \citet{pallavicinipaper}, we propose a linear stochastic volatility model for the fictitious spot price $S_t$, where a general volatility dynamics is considered:
\begin{subequations}\label{eq:model}
	\begin{empheq}[left=\empheqlbrace]{align}
		&dS_t = (\alpha(t) + \beta(t)S_t) \, dt + \sqrt{\xi_t^t} S_t \, dW_t^1, \quad S_0 = \overline{S},\\
		&d\xi_t^u = \lambda(t,u) \, dW_t^2,\label{eq:model_vol}
	\end{empheq}
\end{subequations}
where $\alpha(t)$ is a positive function on time, $\beta(t)$ is a function on time, $\overline{S}$ is the (positive) initial value of the fictitious spot price and $\rho \in [-1,1]$ is the correlation between the standard Brownian motions $W^1$ and $W^2$. 

We consider an affine drift, as it allows for simple calculations to explicitly model the Samuelson effect, which we depict in Section \ref{sec:samuelson}. This methodology has been previously applied to commodity pricing, see for example \citet{drimuspaper, swishchukpaper}.

\begin{hypotheses}\label{hypotheses}
Following the considerations in \citet{bourgeypaper}, we assume that $(\lambda(t,u))_{t \leq u}$ is a scalar stochastic process such that:
\begin{itemize}
    \item $\int_0^u \lambda(t,u)^2 \, dt < \infty$ for every $u\geq0$,
    \item $(\lambda(t,u))_{t \leq u}$ is locally integrable for $t\geq0$,
    \item the initial forward variance curve $\xi_0(u):=\xi_0^u$ is in $L^1_{\text{loc}}(\mathbb{R}^{+})$, and
    \item the equation in (\ref{eq:model_vol}) has a unique weak positive solution and $\xi_t^t$ has almost surely continuous trajectories.
\end{itemize}
\end{hypotheses}

\begin{remark}
    Note that, as shown in \citet{bourgeypaper}, this general forward variance model admits various particular specifications, including the Heston model, the n-factor Bergomi model, and their rough extensions, as will be shown later.
\end{remark}

For convenience, we introduce the \textbf{normalised fictitious spot price}:
\begin{equation*}
    s_t := \dfrac{S_t}{F_0(t)},
\end{equation*}
whose dynamics follows from the product rule and satisfies the following stochastic differential equation:
\begin{equation}\label{normalised_spot_dynamics}
	ds_t = (a(t) + (\beta(t) - \partial_t \ln F_0(t))s_t) \, dt + \sqrt{\xi^t_t} s_t \, dW_t^1,\quad s_0=1,
\end{equation}
where
\begin{equation*}
	a(t):= \dfrac{\alpha(t)}{F_0(t)}.
\end{equation*}

By imposing Equation (\ref{def_fictitious_spot_price}), we establish a relationship between the normalised fictitious spot price $s_t$ and the futures prices $F_t$. In doing so, the functional parameter $\beta(t)$ is automatically calibrated.

First, we need the following lemma:
\begin{lemma}
Given the SDE
\begin{equation*}
dX_t = (a(t) + b(t)X_t) \, dt + g(t) X_t \, dW_t, \quad t\in[0,T],
\end{equation*}
where $X_0$ is given and $a(t), b(t), g(t)$ are continuous functions on $t \in [0,T]$, then the solution is given by
\begin{equation*}
X_t = X_0 Z_t^{-1} + Z_t^{-1} \int_0^t Z_sa(s) \, ds,
\end{equation*}
where
\begin{equation*}
\begin{gathered}
Z_t = \exp\left( \int_0^t \left( \dfrac{1}{2}g^2(s) - b(s) \right)\, ds - \int_0^t g(s) \, dW_s \right),\\
Z_t^{-1} = \exp\left( \int_0^t \left(b(s) - \dfrac{1}{2}g^2(s) \right)\, ds + \int_0^t g(s) \, dW_s \right).
\end{gathered}
\end{equation*}
\end{lemma}

\begin{proof}
This result can be found in any classical book covering stochastic differential equations; see for instance \citet[Section 3.3]{mikoschbook} or \citet[Chapter 4]{ikedabook}.
\end{proof}

Note that Equation (\ref{eq:model_vol}) admits a weak solution by hypothesis; that is, there exists a stochastic process $\xi_t^u$ and a Brownian motion $\tilde{W}_t^2$ such that they satisfy the SDE. We then define a Brownian motion $W_t:=\rho W_t^\perp + \sqrt{1-\rho^2}\tilde{W}_t^2$, where $\rho\in[-1,1]$ and $W_t^\perp$ is a Brownian motion independent of $\tilde{W}_t^2$.

Applying the previous lemma to Equation (\ref{normalised_spot_dynamics}) driven by $W_t$, we obtain
\begin{equation*}
s_t = Z_t^{-1} + Z_t^{-1}\int_{0}^{t} Z_u a(u) \, du,
\end{equation*}
where
\begin{equation*}
Z_t = \exp\left( \int_0^t \left( \dfrac{1}{2}\xi_u^u - \beta(u) + \partial_u \ln F_0(u) \right) \, du - \int_0^t \sqrt{\xi_u^u} \, dW_u \right).
\end{equation*}
Note that
\begin{equation*}
\begin{gathered}
Z_t^{-1} = \exp\left( \int_0^t \mu(u) \, du \right)\mathcal{E}\left( \int_0^t \sqrt{\xi_u^u} \, dW_u \right),
\end{gathered}
\end{equation*}
where $\mu(t):=\beta(t) - \partial_t \ln F_0(t)$ and $\mathcal{E}(\cdot)$ denotes the Doléans-Dade exponential.

Next, consider the process
\begin{equation*}
Y_t := \mathcal{E}\left( \int_0^t \sqrt{\xi_u^u} \, dW_u \right).
\end{equation*}
For now, we will assume that $Y_t$ is a martingale. In the following, we will show that this is the case for some very important specifications of the general forward variance model, such as the rHeston or rBergomi models.

Then, the term $Z_T^{-1}Z_u$ for any $u \in [0,T]$ can be written as:
\begin{equation*}
\begin{gathered}
Z_T^{-1}Z_u = \exp\left( \int_u^T \mu(\omega) \, d\omega \right) \dfrac{\mathcal{E}\left( \int_0^T \sqrt{\xi_\omega^\omega} \, dW_\omega \right)}{\mathcal{E}\left( \int_0^u \sqrt{\xi_\omega^\omega} \, dW_\omega \right)} = \exp\left( \int_u^T \mu(\omega) \, d\omega \right) \dfrac{Y_T}{Y_u},
\end{gathered}
\end{equation*}
with conditional expectation
\begin{equation*}
\mathbb{E}_t[Z_T^{-1}Z_u] = \exp\left( \int_u^T \mu(\omega) \, d\omega \right) \mathbb{E}_t\left[\dfrac{Y_T}{Y_u}\right], \quad \forall t\in[0,T].
\end{equation*}
For a fixed $u\in[0,T]$, it holds that
\begin{equation*}
\mathbb{E}_t[Z_T^{-1}Z_u] =
\begin{cases}
	\exp\left( \int_u^T \mu(\omega) \, d\omega \right)\dfrac{Y_t}{Y_u} \quad &\text{if } t\geq u,\\
	\exp\left( \int_u^T \mu(\omega) \, d\omega \right)\quad &\text{otherwise}.
\end{cases}
\end{equation*}
 Thus, the conditional expectation of $s_T$ is
 \begin{equation}\label{conditional_expectation}
\begin{gathered}
\mathbb{E}_t[s_T] = \exp\left( \int_0^T \mu(\omega) \, d\omega \right) Y_t + \int_0^t  \exp\left( \int_u^T \mu(\omega) \, d\omega \right)\dfrac{Y_t}{Y_u} a(u) \, du + \\ \int_t^T 	\exp\left( \int_u^T \mu(\omega) \, d\omega \right)a(u) \, du,
\end{gathered}
\end{equation}
so that, if we differentiate with respect to $T$, we obtain:
\begin{equation*}
\begin{gathered}
\dfrac{\partial_T F_t(T) F_0(T) - F_t(T)\partial_T F_0(T)}{F_0(T)^2} = \mu(T) \left[ \exp\left( \int_0^T \mu(\omega) \, d\omega \right) Y_t\right. +\\
\left.\int_0^t  \exp\left( \int_u^T \mu(\omega) \, d\omega \right)\dfrac{Y_t}{Y_u} a(u) \, du + \int_t^T \exp\left( \int_u^T \mu(\omega) \, d\omega \right) a(u) \, du \right] + a(T).
\end{gathered}
\end{equation*}

By comparing the previous identity with Equation (\ref{conditional_expectation}), we observe that the term multiplied by $\mu(T)$ is exactly $\mathbb{E}_t[s_T]$. Therefore, we have
\begin{equation*}
\dfrac{\partial_T F_t(T) F_0(T) - F_t(T)\partial_T F_0(T)}{F_0(T)^2} = \mu(T)\mathbb{E}_t[s_T] + a(T) = \mu(T)\dfrac{F_t(T)}{F_0(T)} + a(T).
\end{equation*}
Then, simple algebra leads to
\begin{equation*}
\partial_T F_t(T) = a(T) F_0(T) + \beta(T)F_t(T), \quad F_t(t) = s_t F_0(t),
\end{equation*}
which is an ODE for the futures prices with respect to the maturity date $T>t$.

In particular, for $t=0$ we can solve this equation for
\begin{equation*}
\beta(T) = \partial_T \ln F_0(T) - a(T).
\end{equation*}
Moreover, if we solve the previous ODE, we obtain the following relationship between the normalised fictitious spot price and the futures price:
\begin{equation}\label{futures_prices}
F_t(T) = F_0(T) \left( 1- \left(1-s_t\right)\exp\left(-\int_t^T a(u)\, du\right) \right),
\end{equation}
which leads to the final model:
\begin{equation}\label{model}
	\begin{cases}
		ds_t &= a(t)\cdot(1-s_t) \, dt + \sqrt{\xi_t^t}s_t \, dW_t^1, \quad s_0=1,\\
		d\xi_t^u &= \lambda(t,u) \, dW_t^2,\\
		F_t(T) &= F_0(T) \left( 1- \left(1-s_t\right)\exp\left(-\int_t^{T}a(u)\, du\right) \right),\\
		dW_t^1 dW_t^2 &= \rho \, dt.
	\end{cases}
\end{equation}

The preceding argument is summarised in the following theorem, which allows us to promote models usually employed in equity markets into the commodity framework. Indeed, we  ensure an automatic calibration of the initial term structure of the futures prices and a suitable treatment of the Samuelson effect, as will be discussed at the end of the section.

\begin{theorem}
Given a general forward variance model in the form of (\ref{eq:model_vol}) that satisfies Hypotheses \ref{hypotheses}, the model (\ref{eq:model}) can be written in the form of (\ref{model}) if the stochastic process $Y_t$ defined by
\begin{equation}\label{eq:y_t}
    dY_t = \sqrt{\xi_t^t} Y_t dW_t^1
\end{equation}
is a martingale.
\end{theorem}

The forward variance model (\ref{eq:model_vol}) is very flexible and can take many interesting forms, as shown in \citet{bourgeypaper}. In this work, we are particularly interested in rough volatility models; accordingly, we focus our attention on the rBergomi \citep{bayer} and rHeston models \citep{abijaber}.

\begin{corollary}[rBergomi]\label{cor:rb}
For any function of maturity $F_0(\cdot)$, the system of SDEs
\begin{equation}\label{eq:model_rb}
	\begin{cases}
		ds_t = a(t)\cdot(1-s_t) \, dt + \sqrt{v_t}s_t \, dW_t^1, \quad s_0=1,\\
		v_t = \xi_0(t)\exp\left( \dfrac{-\eta^2}{2}t^{2H} + \eta \sqrt{2H} \displaystyle \int_0^t \dfrac{dW_s^2}{(t-s)^{1/2-H}} \right),
	\end{cases}
\end{equation}
defines a model where future prices given by (\ref{futures_prices}) are martingales starting from $F_0(T)$, where $a(t)>0$ is the mean reversion speed of the fictitious spot price, $\xi_0(t)\geq0$ is the initial forward variance curve, $H\in(0,1)$ is the Hurst parameter, $\eta>0$ is the vol-vol and and $\rho\in[-1,0)$ is the correlation between $W_t^1$ and $W_t^2$.
\end{corollary}

\begin{proof}
    The rBergomi model is a special instance of the general forward variance model (\ref{eq:model_vol}) with
    \begin{equation*}
        \lambda(t,u) = \eta^2 \sqrt{2H} (u-t)^{H-0.5} \xi_t^u, \quad \forall u \geq t,
    \end{equation*}
    see \citet{bourgeypaper} for further information.

    Therefore, we just need to prove that the process $Y_t$, defined in Equation (\ref{eq:y_t}), is a martingale. This follows directly from Theorem 1.1 in \citet{gassiatpaper}. Note that, in order to use this result, it is necessary that $\rho<0$.
\end{proof}

\begin{corollary}[rHeston]\label{cor:rh}
For any function of maturity $F_0(\cdot)$, the system of SDEs
\begin{equation}\label{eq:model_rh}
	\begin{cases}
		ds_t = a(t)\cdot(1-s_t) \, dt + \sqrt{v_t}s_t \, dW_t^1, \quad s_0=1,\\
		v_t = V_0 + \dfrac{1}{\Gamma(H+0.5)} \displaystyle\int_0^t \dfrac{\kappa (\bar{v}(s) - v_s)}{(t-s)^{0.5-H}} \, ds + \dfrac{1}{\Gamma(H+0.5)} \displaystyle\int_0^t \dfrac{\eta \sqrt{v_s}}{(t-s)^{0.5-H}} \, dW_s^2,
	\end{cases}
\end{equation}
defines a model where future prices given by (\ref{futures_prices}) are martingales starting from $F_0(T)$, where $a(t)>0$ is the mean reversion speed of the fictitious spot price, $V_0$ is the initial variance, $H\in(0,1)$ is the Hurst parameter, $\kappa>0$ is the mean reversion speed of the variance, $\eta>0$ is the vol-vol, $\rho\in[-1,1]$ is the correlation between $W_t^1$ and $W_t^2$ and $\bar{v}(t)>0$ is the long variance given as a deterministic function continuous on $\mathbb{R}^+-\{0\}$ satisfying
\begin{equation*}
    \bar{v}(t) \geq \dfrac{-V_0}{\kappa \Gamma(0.5 - H)}t^{-(H+0.5)}, \quad \forall t >0, 
\end{equation*}
and for all $\varepsilon>0$ there exists $K_\varepsilon > 0 $ such that
\begin{equation}\label{longvar_cond}
    \bar{v}(t) \leq K_\varepsilon t^{-0.5-\varepsilon}, \quad \forall t\in(0,1].
\end{equation}
\end{corollary}

\begin{proof}
    The rHeston model can be rewritten as the general model (\ref{eq:model_vol}) with
    \begin{equation*}
        \lambda(t,u) = \eta (u-t)^{H-0.5} f_H\left(-\kappa(u-t)^{H+0.5}\right) \sqrt{\xi_t^t},
    \end{equation*}
    where $\xi_0^u$ is given by the fractional equation
    \begin{equation*}
        \xi_0^u = V_0 + \dfrac{\kappa}{\Gamma(H + 0.5)} \int_0^u (u-s)^{H-0.5}(\bar{v}(s) - \xi_0^s) \, ds,
    \end{equation*}
    $f_H(z) := E_{H+0.5, H+0.5}(z)$ and $E_{\alpha, \beta}(z)$ is the two-parameter Mittag-Leffler function. As previously, see \citet{bourgeypaper} for further information.

    The proof of the martingality of the process $Y_t$, defined in (\ref{eq:y_t}), can be found in \citet[Proposition B.1]{eleuchpaper}. Note that, in contrast to the rBergomi model, in the rHeston model it is not necessary to impose $\rho<0$.
\end{proof}

\begin{remark}
In the previous proof, we rewrote the rHeston volatility model as a forward variance model. The same procedure can be carried out in the opposite direction, provided that the initial forward variance curve $\xi_0$ satisfies certain conditions. Indeed, consider the metric space
\begin{equation*}
    \begin{aligned}
        \mathcal{V}_{H,\kappa} = &\{ \xi:\mathbb{R}^+ \rightarrow \mathbb{R}^+, \quad\exists v_\xi \text{ continuous on } \mathbb{R}^+ \text{ such that}\\
        &\xi(t) = \int_0^t \dfrac{s^{-(H+0.5)}}{\Gamma(0.5 - H)}  (t-s)^{H-0.5} f_H(-\kappa(t-s)^{H-0.5}) v_\xi(s) \, ds\},
    \end{aligned}
\end{equation*}
with $f_H$ as defined in the proof of Corollary \ref{cor:rh}, and its subset
\begin{equation}\label{admissible_fwdvar}
    \begin{aligned}
        \mathcal{V}_{H,\kappa}^+ = &\{ \xi\in\mathcal{V}_{H,\kappa}, \quad v_\xi > 0 \text{ and for any } t>0, \quad \\&v_\xi(t) = \xi(0) + t^{H+0.5} \kappa \Gamma(0.5-H)v_\xi^0(t), \quad v_\xi^0 \text{ satisfies } (\ref{longvar_cond})\},
    \end{aligned}
\end{equation}
which is the set of all admissible forward variance curves produced by the rHeston volatility model \citep[Section 4]{eleuchpaper}.
Given $\xi_0 \in  \mathcal{V}_{H,\kappa}^+$, we can then recover the long variance $\bar{v}$ as
\begin{equation*}
    \kappa\bar{v}(t) +V_0 \dfrac{t^{-(H+0.5)}}{ \Gamma(0.5-H)}  = D^{H+0.5}\xi_0(t) + \xi_0(t), \quad t>0,
\end{equation*}
where $D^rf$ is the Riemann-Lioville fractional derivative of order $r$, as in \citet[Proposition 3.1]{eleuchpaper} and \citet[Section 5.3]{bourgeypaper}.
\end{remark}

Note that we can also consider the classical Heston and 1-factor Bergomi models, since the rough models reduce to their classical counterpart when $H=0.5$.

\subsection{Samuelson effect}\label{sec:samuelson}

    In his seminal paper of 1965, \citet{samuelsonpaper} hypothesised that the volatility of futures price changes should increase when approaching the delivery date, a stylised fact now known as the Samuelson effect. Samuelson did not provide an empirical analysis nor formal proofs, and in the following years several authors analysed the applicability of this hypothesis in several markets, in particular commodities. Several explanations were proposed, and in 1995 Bessembinder et al. \citep{bessembinderpaper} linked the Samuelson effect to the finite inventory of commodities and the mean reversion of spot prices, which has been tested empirically in several commodity markets (such as crude oil, and specially agricultural markets) with positive results, see for instance \citet{hopaper}.

    We remark that the previously cited literature on the Samuelson effect mainly considers the behaviour of the realised volatility, not the implied volatility. Nonetheless, in order to estimate the term structure of futures volatility and price futures contracts effectively, it is reasonable to consider models in which the dynamics of the implied volatility behaves as the Samuelson effect dictates.

    Our model (\ref{model}) is able to reproduce this stylised fact by tuning the mean reversion speed $a(t)$. Indeed, in Figure \ref{fig_samuelson} we represent the relationship between mean reversion speed and the ATM volatility level for several model-generated (that is, non-traded) options on the same underlying futures but with different maturity dates. In absence of mean reversion (that is, for $a=0$) there is no Samuelson effect, but it becomes more prominent for higher mean reversion speeds: the ATM implied volatility level increases when approaching the maturity date. 

    We observe that the ATM volatility level exhibits a downward tendency for mean reversion speeds less than 0.5. This can be related to the seasonality of  commodity prices: for $a=0$, the plot indicates that the calibration data used was retrieved during a period of decreasing volatility, but the Samuelson effect counterforces this trend as the mean reversion speed increases. Figure \ref{fig_samuelson} essentially represents two effects: the Samuelson effect and the effect of seasonality on vanilla option prices.
 
    \begin{remark}
        In order to obtain Figure \ref{fig_samuelson}, we consider market data on futures options with fixed maturities $T_{opt}$ and $T_{fut}$. Specifically, we consider WTI Crude Oil futures options traded on the NYMEX on the 14th of March 2025 with futures ticker CLU5, as described in Table \ref{tab:1st}. Next, we calibrate the model (\ref{eq:model_rb}) using a fixed mean reversion speed and compute the ATM implied volatility of model-generated options for different values of $T_{opt}$.
    \end{remark}

\begin{figure}%
    \centering
	\includegraphics[width=0.55\textwidth]{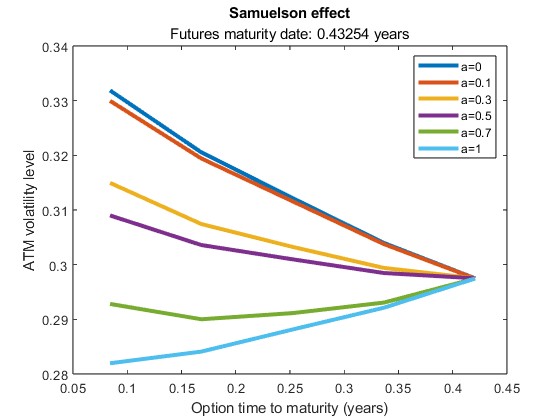}
	\caption{ATM implied volatility with respect to the options maturity date of a simulated mid-curve option for different mean reversion speeds. The Samuelson effect is more pronounced as the mean reversion speed increases.}\label{fig_samuelson}%
\end{figure}

\section{Simulation and calibration}\label{sec:calibration}

From now on, we will consider the rBergomi  (\ref{eq:model_rb}) and the rHeston (\ref{eq:model_rh}) models simulated using a Monte Carlo scheme. The dynamics of the spot price $s_t$ can be simulated using the Euler-Maruyama scheme. However, the simulation of the variance process is not as direct due to the presence of the fractional Brownian motion.

In this work, we use the Hybrid Scheme proposed in \citet{pakkanenpaper} to simulate the rBergomi model efficiently, and the Hybrid Quadratic Scheme introduced by Gatheral in \citet{hqepaper} to simulate the rHeston model, both of which have been implemented in MATLAB R2024b. We refer the reader to the aforementioned references \citet{pakkanenpaper} and \citet{hqepaper} for a more detailed discussion.

\subsection{Initial forward variance curve model}\label{sec:xi}
The initial forward variance curve $\xi_0(t)$ must be known prior to simulation of the rBergomi (\ref{eq:model_rb}) and rHeston (\ref{eq:model_rh}) models. In this work, we assume that this curve follows a given parametric model and calibrate its parameters.

Several different models were considered, such as piecewise constant, piecewise linear, and the parametric models proposed in \citet{buehlerpaper}. We conducted several tests and conclude that the qualitative results were similar regardless of the specific initial forward variance curve model. Based on these results, we consider a piecewise linear initial forward variance curve. This continuous model offers enough flexibility to accommodate any reasonable setting, while the straightforward interpretation of its parameters enables efficient sequential calibration.

Note that the set $\mathcal{V}_{H,\kappa}^+$ of admissible initial forward variance curves for the rHeston volatility model, defined in (\ref{admissible_fwdvar}), does not contain piecewise linear functions except purely linear ones. Nevertheless, for the practical reasons discussed above, we consider a piecewise linear initial forward variance curve. This choice is consistent with the empirical literature, where non admissible forward variance curves, such as piecewise constant curves, are commonly employed, see for instance \citet{baschettipaper}. Since the rHeston volatility model cannot be recovered from these formulations, these may be interpreted as practical approximations of the original model.

\subsection{Mean reversion speed $a(t)$}\label{sec:mrs}
It is not possible to calibrate the mean reversion speed $a$ (which we will assume constant from now on) using only vanilla options, since these instruments are not sensitive enough to changes in this parameter. Indeed, in our tests we were able to calibrate model (\ref{eq:model_rb}) for different fixed values of $a$ between $0.3$ and $0.7$.

This limitation was already present in the model on which our framework is based (see \citet[Section 4.4]{pallavicinipaper}), where the authors addressed it by enriching the calibration set, that is, by considering other kinds of options contracts (such as mid-curve options) alongside vanilla options. In this work, we choose to fix the mean reversion speed at $a=0.5$, which is a value confirmed to be reasonable through our experiments. Although not ideal, this approach produces satisfactory results when calibrating to vanilla options, as shown in Section \ref{sec:results}. A more systematic methodology for calibrating $a$ remains an open question for future investigation.

\subsection{Calibration scheme}
We calibrate the general model (\ref{model}) to market data of vanilla options on WTI Crude Oil futures following a nested calibration procedure. We calibrate the constant parameters $\Theta$ of the model (that is, all parameters except the mean reversion $a$ and the initial forward variance curve $\xi_0(t)$, whose calibration we will explain shortly later on) by minimising the following cost function:
\begin{equation}
		\begin{aligned}
                &L(\Theta) = \sum_i L_i(\Theta),\\
				&L_i(\Theta) = \dfrac{1}{\omega_i} \sum_{j\in\mathcal{J}_i} \omega_{i,j} \left| \sigma_{i,j}^{mkt} - \sigma_{i,j}^{m} \right| + \sum_{j\in\mathcal{J}_i} \mathbf{1}_{\left\{| \sigma_{i,j}^{mkt} - \sigma_{i,j}^{m} | > 0.03\right\}} \left| \sigma_{i,j}^{mkt} - \sigma_{i,j}^{m} \right|,\\
				&\omega_{i,j} := \dfrac{vol_{i,j}}{\max(0.01, BA_{i,j})}, \quad \omega_i := \sum_{j\in\mathcal{J}_i} \omega_{i,j},
			\end{aligned}
    \label{cost1}
\end{equation}
where the index $i$ cycles through all maturities considered for the calibration, and the index $j\in\mathcal{J}_i$ through all the available strike prices for maturity $i$; $vol_{i,j}$ is the trading volume of the contract with maturity $i$ and strike $j$, $BA_{i,j}$ is the corresponding bid-ask spread, $\sigma^{mkt}_{i,j}$ denotes its quoted implied volatility, and $\sigma^{m}_{i,j}$ the implied volatility calculated from the prices generated by the model. The symbol $\mathbf{1}$ denotes the indicator function.
    
For each evaluation of the cost function (\ref{cost1}), we calibrate the levels of the initial forward variance curve $\xi_0(t)$ sequentially using the bisection method, that is, one level per maturity at a time, by minimising the distance between the model-generated volatility and the quoted volatility at the ATM strike. This choice is justified by the fact that different forward variance curves primarily affect the volatility smiles as a vertical shift. By exploiting this property, we are able to to drastically reduce the dimensionality of the optimisation space.

The purpose of the second term in the loss function $L$ defined in (\ref{cost1}) is to penalise volatility smiles that achieve very high accuracy near the ATM strike but exhibit large errors in the ITM or OTM regions. Without this term, such smiles would obtain an excessively low loss value, since the trading volume of ATM options is considerably larger, as will be shown in the next section. This penalty term is only applied when the distance between the model-generated and quoted implied volatility exceeds 0.03. After testing possible cut-off values, we found that this value achieves the previously pointed objective, and leads to the calibration results obtained in Section \ref{sec:results}. More systematic approaches could be considered, for example, allowing the cut-off value to depend on the bid-ask spread.
    
Different cost functions, such as MSE or relative error, were also considered and we did not find any significant qualitative differences.

We solve the optimisation problem defined by the cost function (\ref{cost1}) with the surrogate model optimiser \texttt{surrogateopt} provided by the Global Optimization Toolbox in MATLAB 2024b. Afterwards, we use the local optimiser \texttt{fmincon} to further improve the precision of the estimation.
    
\section{Numerical results}\label{sec:results}

\subsection{Estimation of the Hurst parameter}

In this section, we estimate the Hurst parameter of WTI Crude Oil futures prices traded on the NYMEX. We will proceed with the methodology presented in \citet{gatheralpaper}, which we summarise below.

Given $q,\Delta>0$, we use daily realised volatility proxies from 5-min returns at times $t=0,\Delta,2\Delta,\dots,N\Delta$, where $N:=\lfloor T/\Delta \rfloor$, to calculate
\begin{equation*}
    m(q, \Delta) := \dfrac{1}{N} \sum_{k=1}^N |\log RV_{k\Delta} - \log RV_{(k-1)\Delta}|^q.
\end{equation*}

Assuming that the increments of the log-volatility process are stationary and that a law of large numbers holds, then the quantity $m(q,\Delta)$ can be seen as the empirical counterpart of
\begin{equation*}
    \mathbb{E}\left[ |\log RV_{\Delta} - \log RV_0|^q \right],
\end{equation*}
and the log-volatility increments enjoy the following property:
\begin{equation*}
    \mathbb{E}\left[ |\log RV_{\Delta} - \log RV_0|^q \right] = K_q \Delta^{Hq}.
\end{equation*}

Thus, it is reasonable to assume that $m(q,\Delta) \propto \Delta^{Hq}$ and we can estimate $H$ by means of a linear regression of $\log m(q,\Delta)$ on $\log \Delta$ for several values of $\Delta$.

\begin{table}[!t]
\centering
\begin{tabular}{cccc}
\toprule
ID   & Maturity Month & \begin{tabular}[c]{@{}l@{}}Raw data count\\ (10-sec returns)\end{tabular} & \begin{tabular}[c]{@{}l@{}}Clean data count\\ (5-min returns)\end{tabular} \\ 
\midrule
CLQ5 & August 2025    & 365962 & 54102 \\ 
CLU5 & September 2025 & 318379& 56623  \\
CLV5 & October 2025   & 154334  & 56613 \\
CLX5 & November 2025  & 83344 & 56624 \\
CLZ5 & December 2025  & 341052 & 56624 \\
CLM6 & June 2026      & 95027 & 56602 \\
CLZ6 & December 2026  & 143435 & 56623 \\
\bottomrule
\end{tabular}
\caption{Futures contracts on WTI Crude Oil traded in the NYMEX, considered for the estimation of the Hurst parameter.}\label{tab:hurst}
\end{table}

We consider market data on WTI Crude Oil futures traded on the NYMEX from 16th January 2025 to 31st July 2025. This data set consists of price observations for seven different futures contracts, aggregated into 10-second bins. After cleaning the data, 5-min returns were used to construct daily proxies for realised volatility, in order to mitigate microstructure noise\footnote{2-min and 10-min returns were also tested and no qualitative difference was found.}. Table \ref{tab:hurst} collects the main characteristics of the considered futures contracts.
\begin{figure}[h!]%
	\centering
	\subfloat{
		{\includegraphics[width=0.45\textwidth]{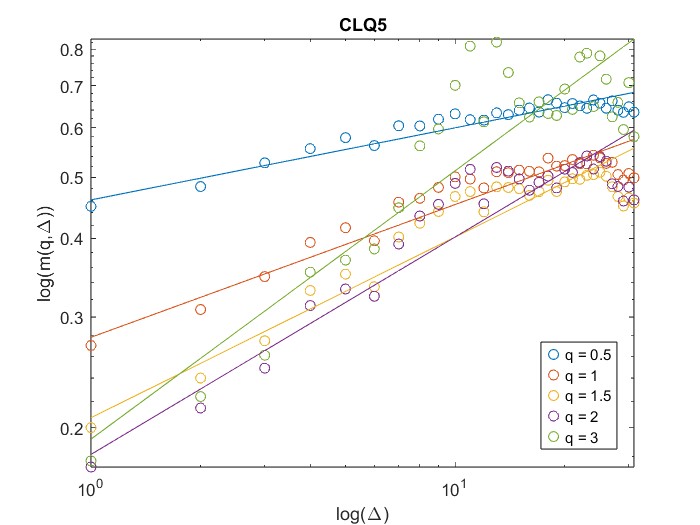} }
	}%
	\subfloat{
		{\includegraphics[width=0.45\textwidth]{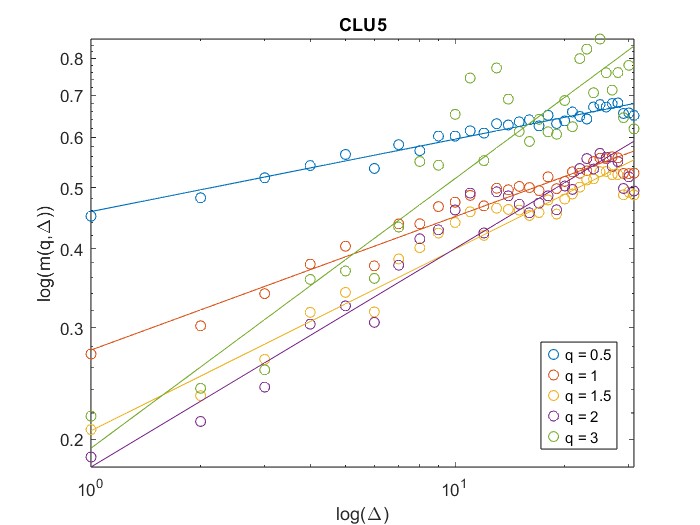} }
	}%
	
	\subfloat{
		{\includegraphics[width=0.45\textwidth]{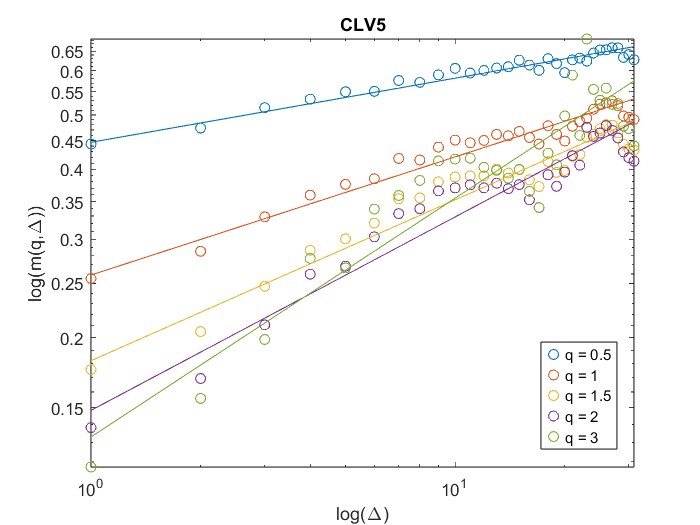} }
	}%
	\subfloat{
		{\includegraphics[width=0.45\textwidth]{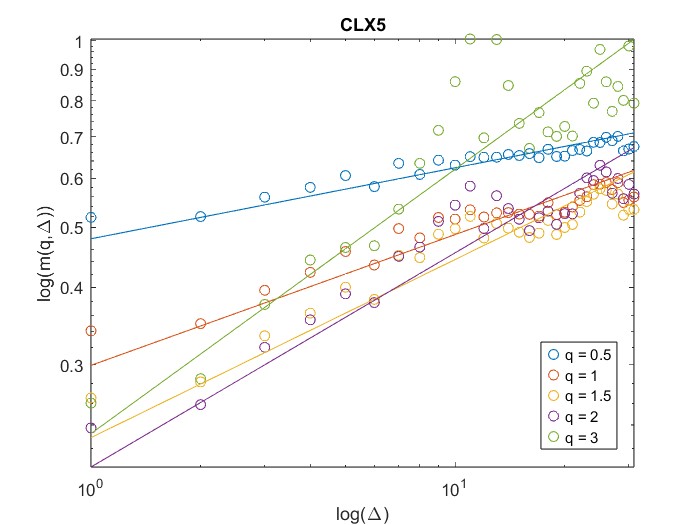} }
	}%
	
	\subfloat{
		{\includegraphics[width=0.45\textwidth]{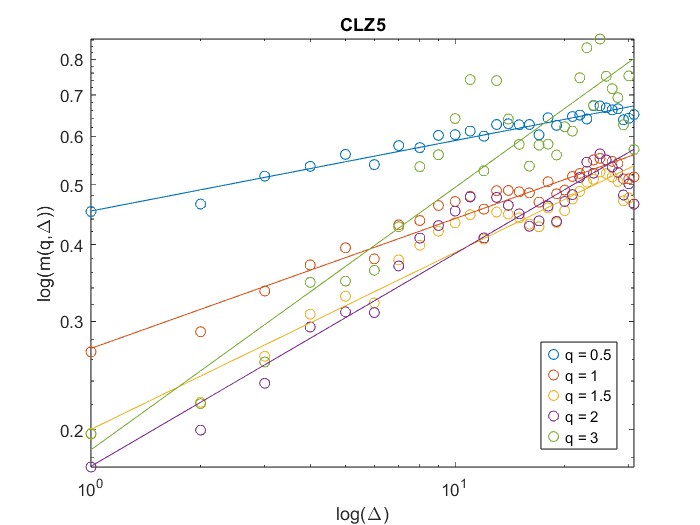} }
	}%
	\subfloat{
		{\includegraphics[width=0.45\textwidth]{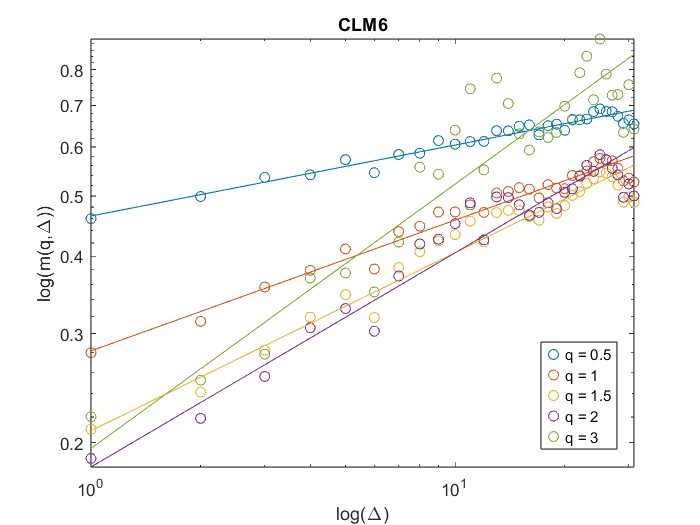} }
	}%
	
	\subfloat{
		{\includegraphics[width=0.45\textwidth]{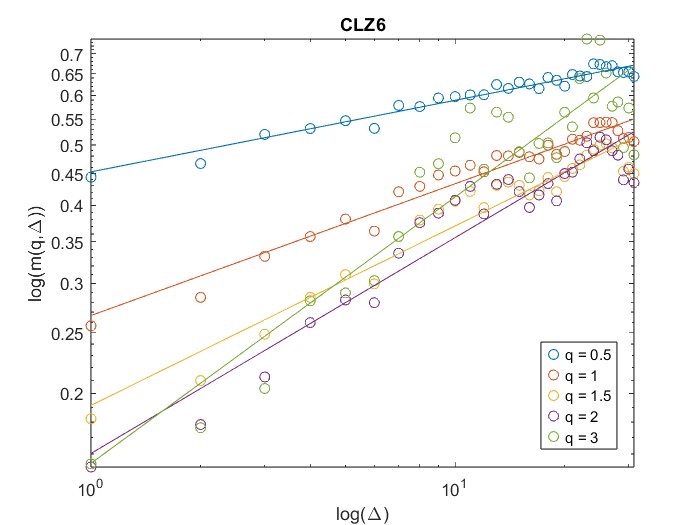} }
	}%
	\caption{linear regression of $\log m(q,\Delta)$ on $\log \Delta$ for the WTI Crude Oil futures contracts listed in Table \ref{tab:hurst}.}%
	\label{fig:linear_reg}%
\end{figure}
The resulting data set is considerably smaller than the one used in the original paper \citep{gatheralpaper} due to the idiosyncrasies of the WTI Crude Oil market. In order to enlarge our effective sample, we assume a single common Hurst parameter $H$ for all contracts considered. This assumption is supported empirically, since the estimated Hurst parameters for each individual futures contract were found to be very similar.

Following the suggestion in \citet{gatheralpaper}, we compute $m(q,\Delta)$ for $q\in\{0.5,1,1.5,2,3\}$ and $\Delta = \{1,2,\dots,31\}$. The resulting linear regression can be found in Figure \ref{fig:linear_reg}, and the estimated Hurst parameter ranges from $0.14$ to $0.21$ depending on the value of $q$. We note that our results are considerably noisier than those reported in \citet{gatheralpaper}: this could be attributed to the limited size of our data set. Nonetheless, we also remark that our results show a clear presence of rough behaviour in the WTI Crude Oil market.

\citet{alfeuspaper} also studied the Hurst parameter of several commodities and found that the Hurst parameter for all the considered commodities was less than 0.2, that is, these commodities prices showed rough behaviour. In particular,  they obtained $H=0.153$ for WTI Crude Oil prices, which is consistent with the results we obtain and present in this section. Nonetheless, we consider futures contracts with different maturities as distinct assets, and then constrain the Hurst parameter $H$ to be the same for all. In contrast, Alfeus et al. performed a single regression on the time series of the front futures contract.

\subsection{Calibration data}

We calibrate the rBergomi and rHeston models to market data of futures options on WTI Crude Oil traded on the NYMEX (New York Mercantile Exchange) on 14th March 2025 and 4th June 2025. The maturity dates of the contracts within each data set are collected in Table \ref{tab:1st} and \ref{tab:2nd}. The data includes bid and ask prices for both call and put options across a range of strikes and maturities, as well as the implied volatility and trading volume for each contract. Not all available contracts were considered: instead, we used the volume data to select the most reliable and liquid contracts.

\begin{table}[!hb]
			\centering
			\begin{tabular}{ccc}
						\toprule
						Futures ticker & $T_{opt}$ & $T_{fut}$\\
                        \midrule
                        CLJ5 & 17th March 2025 & 20th March 2025\\ 
						CLK5 & 16th April 2025 & 22nd April 2025\\ 
						CLM5 & 15th May 2025 & 20th May 2025\\ 
						CLN5 & 16th June 2025 & 20th June 2025\\ 
						CLQ5 & 17th July 2025 & 22nd July 2025\\ 
						CLU5 & 15th August 2025 & 20th August 2025\\ 
						CLZ5 & 17th November 2025 & 20th November 2025\\
                        \bottomrule
			\end{tabular}
            \caption{Maturity dates of the WTI Crude Oil futures option contracts considered and their respective underlying futures contract, traded on NYMEX on the 14th of March 2025.}\label{tab:1st}
\end{table}
\begin{table}[!ht]
	\centering
	\begin{tabular}{ccc}
		\toprule
		Futures ticker & $T_{opt}$ & $T_{fut}$\\
        \midrule
		CLN5 & 16th June 2025 & 20th June 2025\\
		CLQ5 & 17th July 2025 & 22nd July 2025\\
		CLU5 & 15th August 2025 & 20th August 2025\\
		CLV5 & 17th September 2025 & 22nd September 2025\\
		CLX5 & 16th October 2025 & 21st October 2025\\
		CLZ5 & 17th November 2025 & 20th November 2025\\
        \bottomrule
	\end{tabular}
    \caption{Maturity dates of the WTI Crude Oil futures option contracts considered and their respective underlying futures contract, traded on NYMEX on the 4th of June 2025.}\label{tab:2nd}
\end{table}

It is important to note that only American options are quoted in these data sets. For a given strike and maturity date, the prices of an American and European option generally differ, hence it would be necessary to implement an algorithm to price American options under the considered models. Nevertheless, for calibration purposes, we assume that the quoted implied volatility for American options is also valid for European options. In order to assess the validity of this assumption, we computed the implied volatility smiles obtained by inverting the American market prices using the Black-Scholes formula for European options. For both considered data sets, the difference between the resulting European volatility smiles and the quoted American smiles is empirically negligible compared to bid-ask spreads, thereby supporting the validity of our assumption.

\subsection{Calibration results: 14th March 2025}\label{sec:dataset1}

We now examine the results obtained from calibrating  the rBergomi  and rHeston models to the data set summarised in Table \ref{tab:1st}. For comparison, we also report the results obtained by calibrating the model (\ref{eq:model}) with both classical Bergomi and Heston models, that is,
\begin{equation}\label{eq:model_b}
	\begin{cases}
		ds_t = a\cdot(1-s_t) \, dt + \sqrt{v_t}s_t \, dW_t^1, \quad s_0=1,\\
		v_t = \xi_0(t)\exp\left( \eta X_t - \dfrac{1}{2}\eta^2 \texttt{Var}(X_t) \right),\\
        dX_t = -\kappa X_t \, dt + dW_t^2, \quad X_0=0,
	\end{cases}
\end{equation}
and
\begin{equation}\label{eq:model_h}
	\begin{cases}
		ds_t = a\cdot(1-s_t) \, dt + \sqrt{v_t}s_t \, dW_t^1, \quad s_0=1,\\
		dv_t = \kappa (\bar{v}(t) - v_t) \, dt + \eta\sqrt{v_t} \, dW_t^2, \quad v_0 = \hat{v}_0
	\end{cases}
\end{equation}
respectively, where $\kappa>0$ is the mean reversion speed, $\eta$ is the vol-vol, $\xi_0(t)$ is the initial forward variance curve, $\bar{v}(t)$ is the long-term variance and $\hat{v}_0$ the initial variance for the Heston model. We simulate and calibrate these models using an Euler-Maruyama scheme, and assume that $\bar{v}(t)$ is given by a piecewise linear function, which we treat in a similar manner to the initial forward variance curve described in Section \ref{sec:xi}. Since all other models considered feature a time-dependent parameter, we naturally extend the classical Heston model by introducing a time-dependent long-term variance.

Throughout the remainder of the tests, $N$ will denote the number of trajectories simulated and $n$ the inverse of the step size.

We note that the maturity date of the first contract lies only two days after the observation date. Therefore, in order to simulate this contract accurately, we would need to choose a very small time step size $1/n$, which would increase the computational time significantly. While this is not problematic for the classical models, the rough Bergomi and rough Heston models are substantially more computationally demanding, and thus a direct approach is impractical. In order to overcome this issue, we use a dual mesh approach: we simulate the first contract on a very fine mesh with $n=2000$ and the rest of the contracts on a coarser mesh with $n=300$. We consider a single mesh with $n=2000$ for the classical models, and simulate $N=100000$ trajectories for all the models considered.

\begin{remark}
    One possible way to avoid the dual mesh approach and speed up the calculations significantly would be to approximate the price of plain vanilla options by neural networks \citep{baschettipaper}. We are aware of this possibility, however we do not consider this relevant to the purpose of the paper, namely comparing different models.
\end{remark}

The optimised time-independent parameters are shown in Table \ref{tab:res_param1}. In Figure \ref{fig:res_param1} we can observe the time-dependent parameters of each model, namely the initial forward variance curve $\xi_0(t)$ for the Bergomi, rBergomi and rHeston models and the long-term variance $\bar{v}(t)$ for the Heston model.

\begin{table}[!htbp]
\centering
\begin{tabular}{lccccc}
\toprule
Volatility model   & $H$    & $\eta$  & $\rho$  & $\kappa$ & $\hat{v}_0$ \\
\midrule
rBergomi         & 0.0778 & 2.1617  & -0.3087 & --       & --          \\
rHeston          & 0.2774 & 2.0567  & -0.2017 & 5.6187   & --          \\
1-factor Bergomi & --     & 16.4983 & -0.2108 & 46.6008  & --          \\
Heston           & --     & 9.9747  & -0.2004 & 42.8659  & 0.0405      \\
\bottomrule
\end{tabular}
\caption{Optimised time-independent parameters of the models calibrated to the data from 14th March 2025.}\label{tab:res_param1}
\end{table}

\begin{figure}[ht!]%
	\centering
	\subfloat[rBergomi: $\xi_0(t)$]{
		{\includegraphics[width=0.45\textwidth]{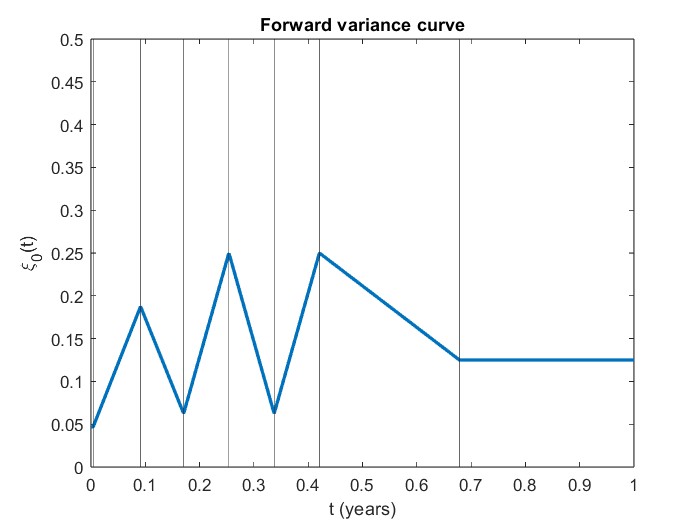} }
	}%
	\subfloat[rHeston: $\xi_0(t)$]{
		{\includegraphics[width=0.45\textwidth]{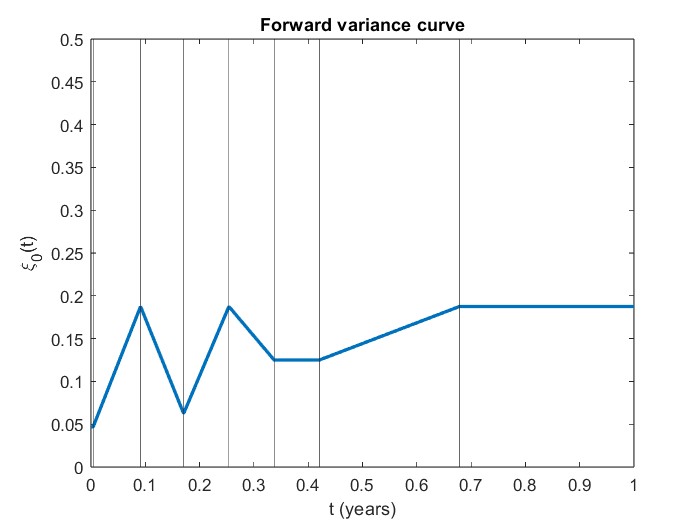} }
	}%
	
	\subfloat[Bergomi: $\xi_0(t)$]{
		{\includegraphics[width=0.45\textwidth]{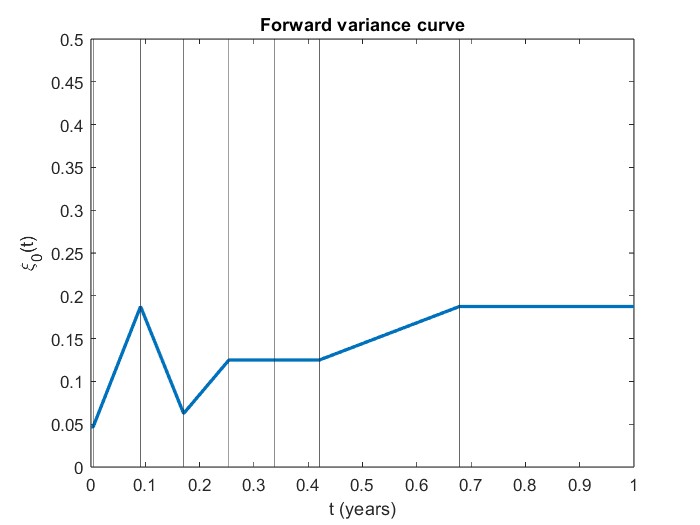} }
	}%
	\subfloat[Heston: $\bar{v}(t)$]{
		{\includegraphics[width=0.45\textwidth]{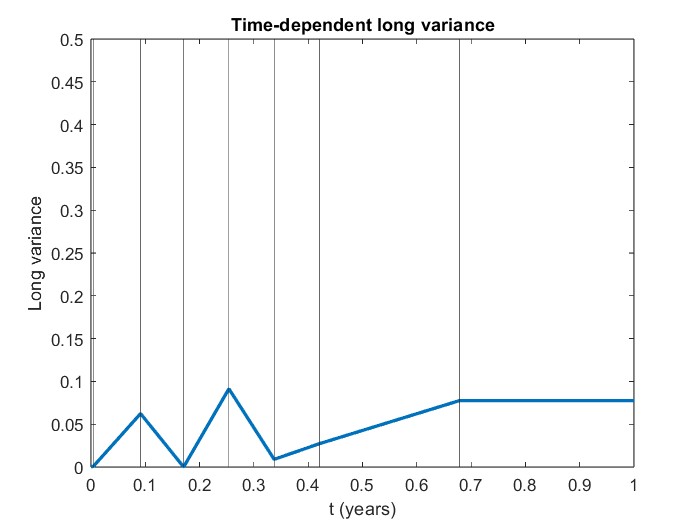} }
	}%
	
	\caption{Optimised time-dependent parameter of each model calibrated to the data from 14th March 2025. The vertical lines mark the maturity dates considered.}%
	\label{fig:res_param1}%
\end{figure}


In Figures \ref{fig:smile_rb1} to \ref{fig:smile_h1} we can observe the volatility smiles generated by the models and compare them to the market smiles. We also plot the trading volume, which is represented as a green bar plot in the background, and indicate the loss value $L_i(\Theta_{\text{opt}})$ obtained from each single smile.

\begin{figure}[!p]%
	\centering
	\subfloat[Loss $= 0.0780$]{
		{\includegraphics[width=0.45\textwidth]{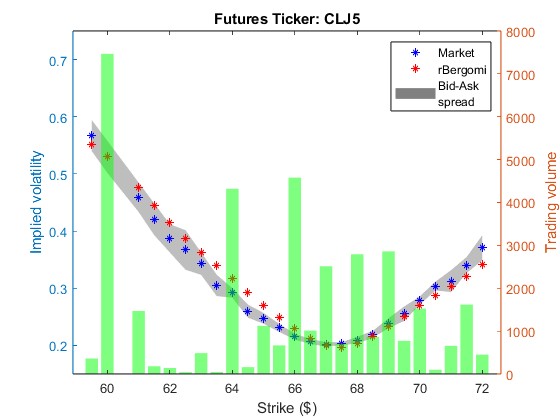} }
	}%
	\subfloat[Loss $= 0.0077$]{
		{\includegraphics[width=0.45\textwidth]{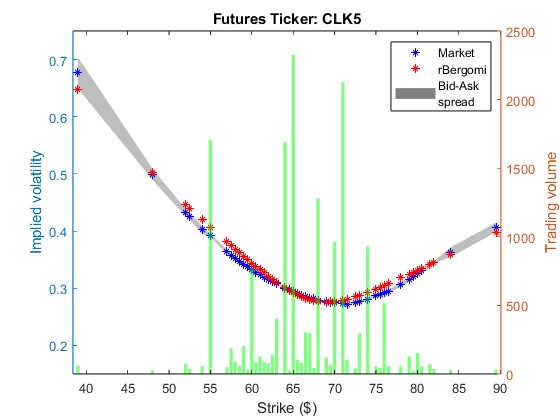} }
	}%
	
	\subfloat[Loss $= 0.0131$]{
		{\includegraphics[width=0.45\textwidth]{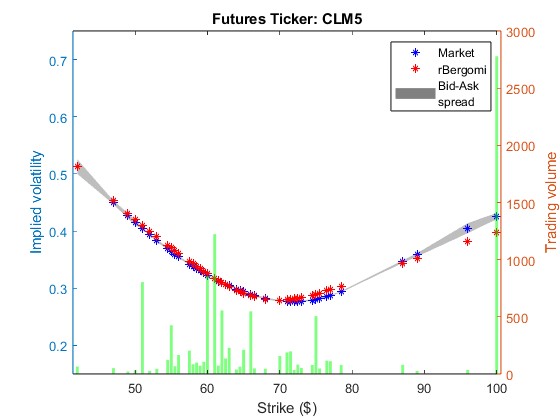} }
	}%
	\subfloat[Loss $= 0.0061$]{
		{\includegraphics[width=0.45\textwidth]{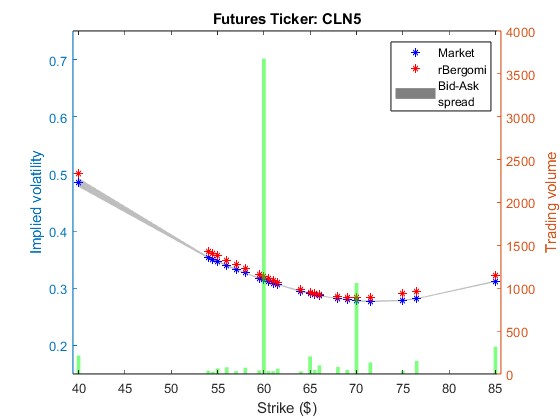} }
	}%
	
	\subfloat[Loss $= 0.0078$]{
		{\includegraphics[width=0.45\textwidth]{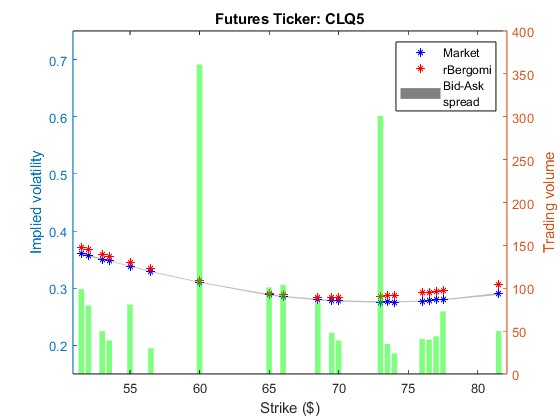} }
	}%
	\subfloat[Loss $= 0.0033$]{
		{\includegraphics[width=0.45\textwidth]{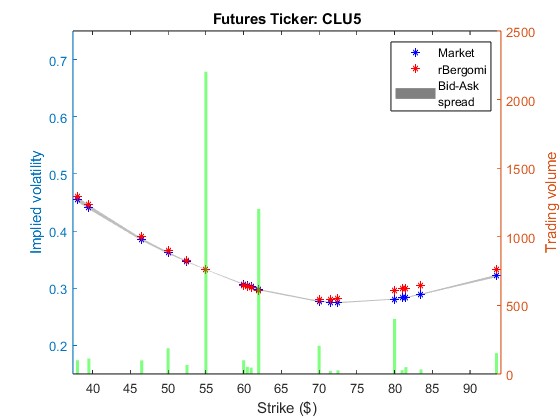} }
	}%
	
	\subfloat[Loss $= 0.0095$]{
		{\includegraphics[width=0.45\textwidth]{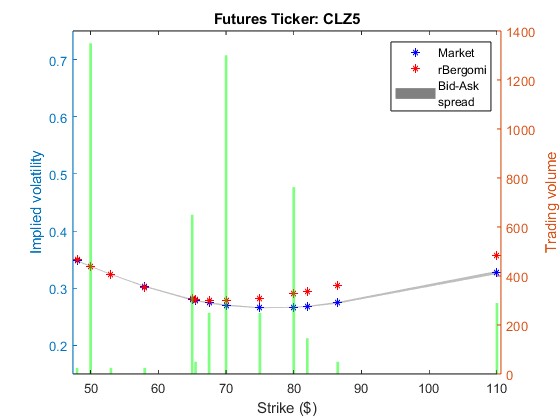} }
	}%
	\caption{Calibration results of the rBergomi model (\ref{eq:model_rb}) for the data from 14th March 2025.}%
	\label{fig:smile_rb1}%
\end{figure}


\begin{figure}[!p]%
	\centering
	\subfloat[Loss $= 0.3184$]{
		{\includegraphics[width=0.45\textwidth]{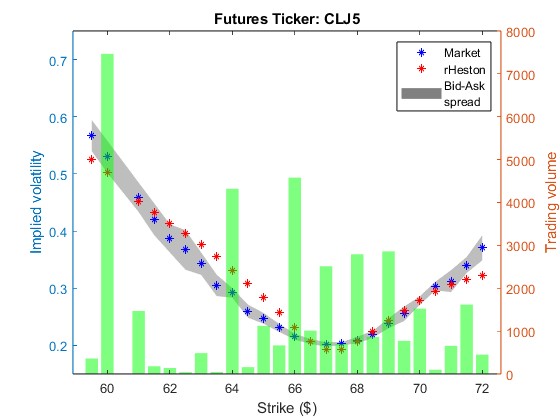} }
	}%
	\subfloat[Loss $= 0.0631$]{
		{\includegraphics[width=0.45\textwidth]{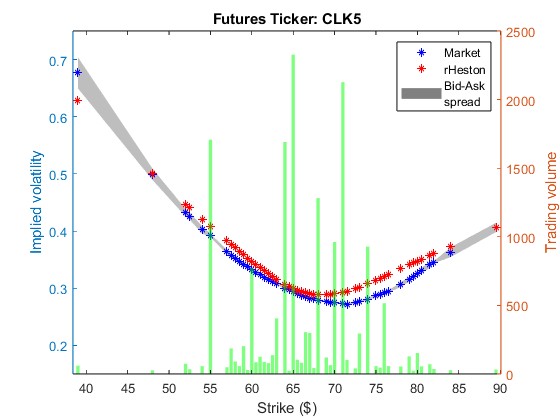} }
	}%
	
	\subfloat[Loss $= 0.0061$]{
		{\includegraphics[width=0.45\textwidth]{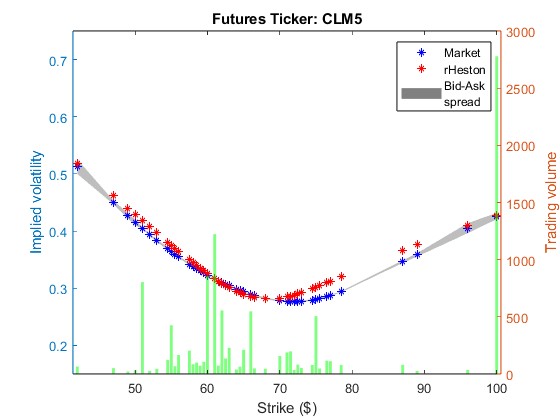} }
	}%
	\subfloat[Loss $= 0.0058$]{
		{\includegraphics[width=0.45\textwidth]{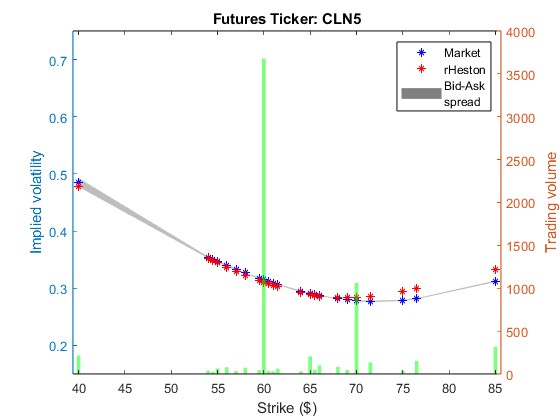} }
	}%
	
	\subfloat[Loss $= 0.0086$]{
		{\includegraphics[width=0.45\textwidth]{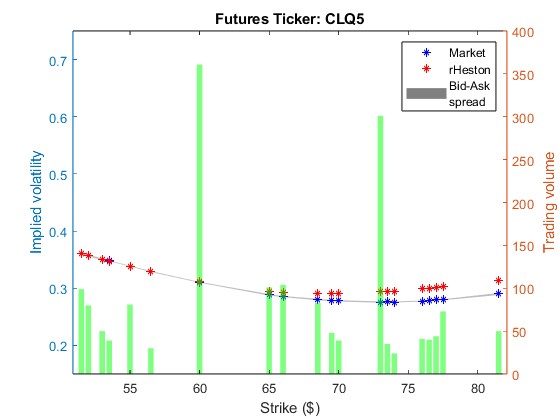} }
	}%
	\subfloat[Loss $= 0.0104$]{
		{\includegraphics[width=0.45\textwidth]{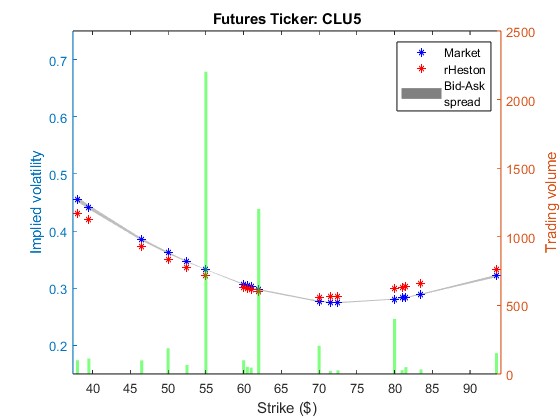} }
	}%
	
	\subfloat[Loss $= 0.0165$]{
		{\includegraphics[width=0.45\textwidth]{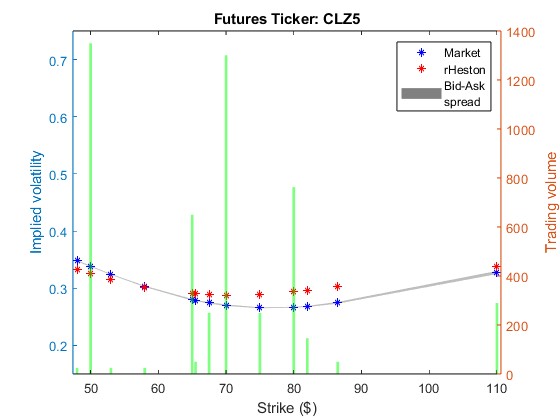} }
	}%
	\caption{Calibration results of the rHeston model (\ref{eq:model_rh}) for the data from 14th March 2025.}%
	\label{fig:smile_rh1}%
\end{figure}


\begin{figure}[!p]%
	\centering
	\subfloat[Loss $= 0.3308$]{
		{\includegraphics[width=0.45\textwidth]{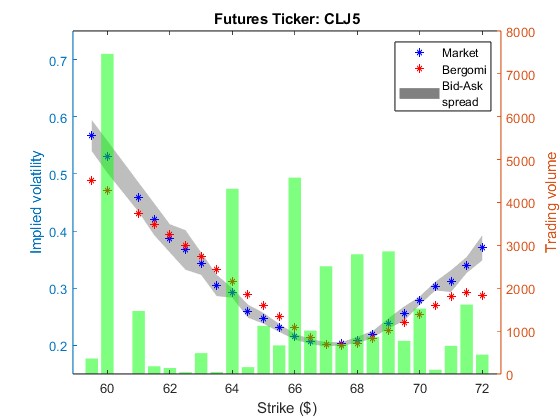} }
	}%
	\subfloat[Loss $= 0.3640$]{
		{\includegraphics[width=0.45\textwidth]{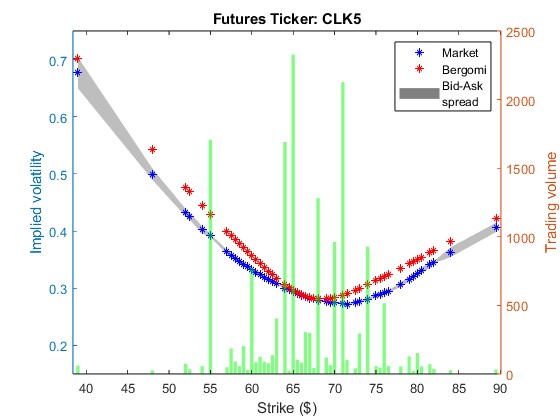} }
	}%
	
	\subfloat[Loss $= 0.0139$]{
		{\includegraphics[width=0.45\textwidth]{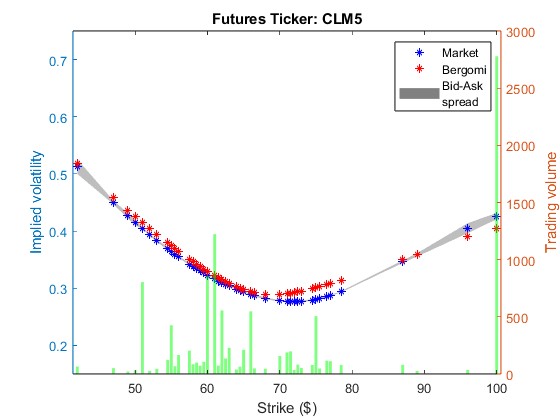} }
	}%
	\subfloat[Loss $= 0.0086$]{
		{\includegraphics[width=0.45\textwidth]{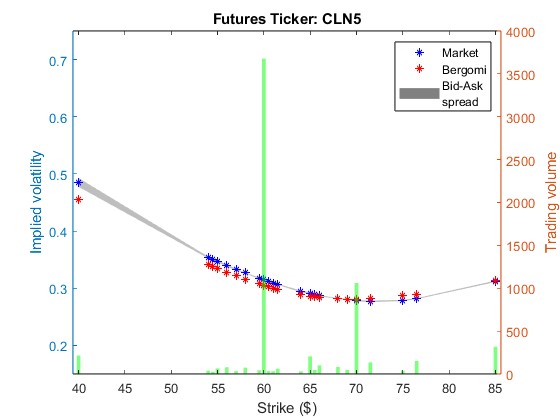} }
	}%
	
	\subfloat[Loss $= 0.0116$]{
		{\includegraphics[width=0.45\textwidth]{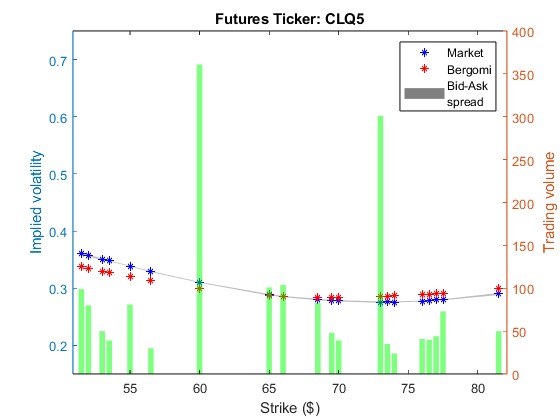} }
	}%
	\subfloat[Loss $= 0.2110$]{
		{\includegraphics[width=0.45\textwidth]{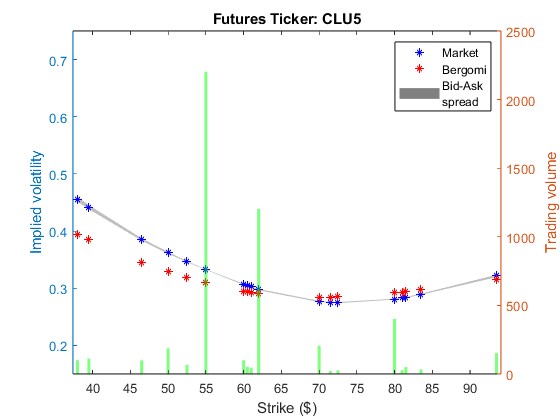} }
	}%
	
	\subfloat[Loss $= 0.0542$]{
		{\includegraphics[width=0.45\textwidth]{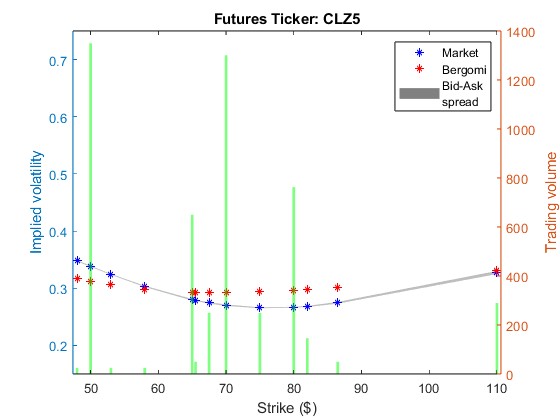} }
	}%
	\caption{Calibration results of the 1-factor Bergomi model (\ref{eq:model_b}) for the data from 14th March 2025.}%
	\label{fig:smile_b1}%
\end{figure}


\begin{figure}[!p]%
	\centering
	\subfloat[Loss $= 0.5015$]{
		{\includegraphics[width=0.45\textwidth]{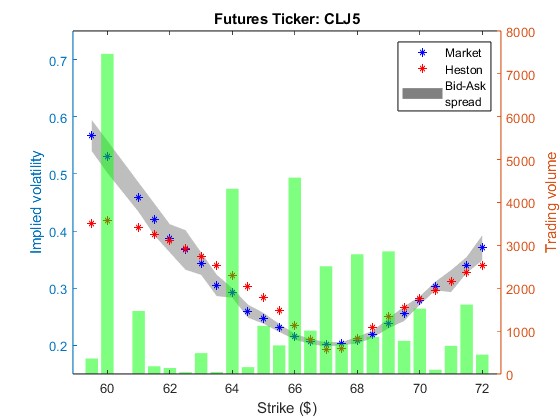} }
	}%
	\subfloat[Loss $= 0.0535$]{
		{\includegraphics[width=0.45\textwidth]{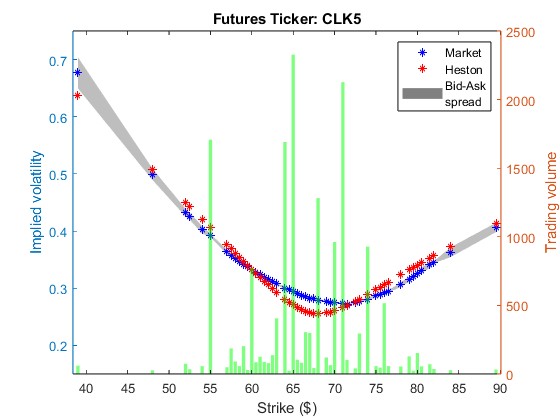} }
	}%
	
	\subfloat[Loss $= 0.0054$]{
		{\includegraphics[width=0.45\textwidth]{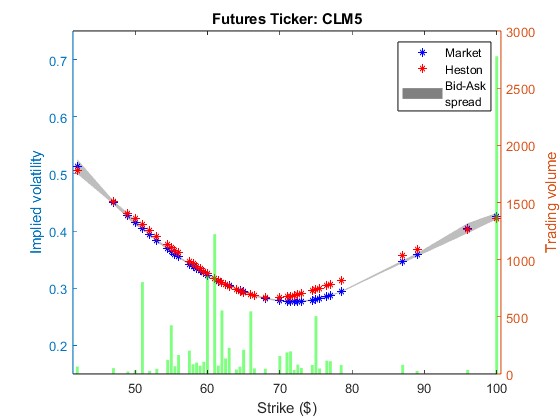} }
	}%
	\subfloat[Loss $= 0.0043$]{
		{\includegraphics[width=0.45\textwidth]{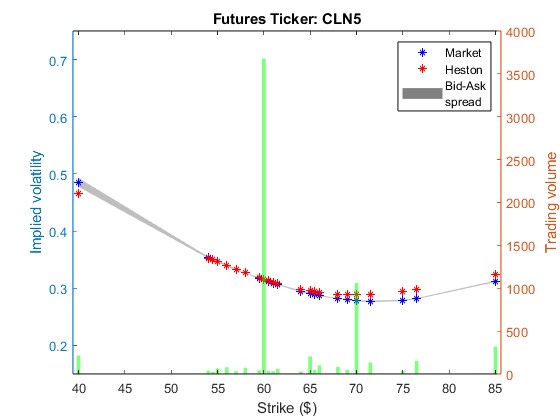} }
	}%
	
	\subfloat[Loss $= 0.0152$]{
		{\includegraphics[width=0.45\textwidth]{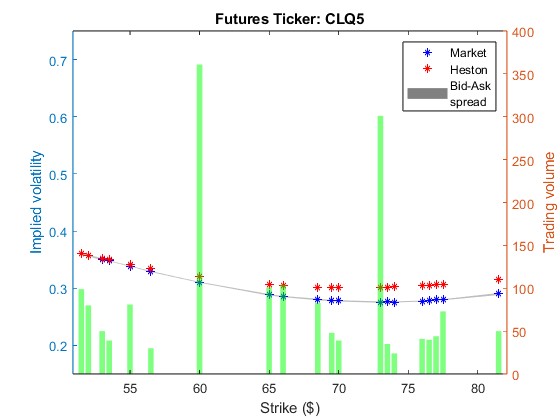} }
	}%
	\subfloat[Loss $= 0.0912$]{
		{\includegraphics[width=0.45\textwidth]{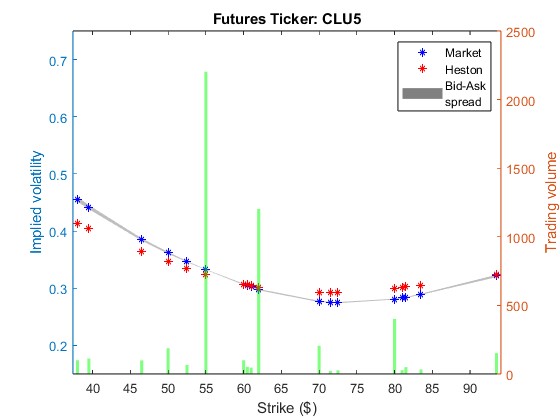} }
	}%
	
	\subfloat[Loss $= 0.0200$]{
		{\includegraphics[width=0.45\textwidth]{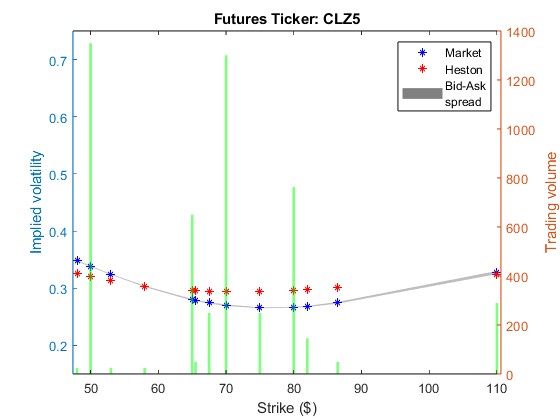} }
	}%
	\caption{Calibration results of the Heston model (\ref{eq:model_h}) for the data from 14th March 2025.}%
	\label{fig:smile_h1}%
\end{figure}

Rough models yield better quantitative results than their classical counterparts, although the qualitative differences remain negligible. In order to achieve these results, the calibrated classical models exhibit very high vol-vol $\eta$ and mean reversion speed $\kappa$, as shown in Table \ref{tab:res_param1}. We also remark on the different Hurst parameters $H$ obtained when calibrating the rBergomi and rHeston models, $0.0778$ and $0.2774$, respectively.

In order to obtain a broader perspective on these issues, we repeat the calibration using the second data set; for which the corresponding results are presented in the following section.

\begin{remark}
    The results of the classical Bergomi and Heston models seem to exhibit slight numerical instabilities in the very first smile. This could be attributed to the simulation scheme used: we decided to use the Euler-Maruyama scheme for its simplicity and ease of implementation, but these numerical instabilities could likely be avoided considering alternative schemes.
\end{remark}

\FloatBarrier

\subsection{Calibration results: 4th June 2025}
As Table \ref{tab:2nd} shows, the shortest time to maturity in the second data set is 12 days. Therefore, unlike in the previous data set, we do not need to use the dual mesh approach to simulate the rBergomi model (\ref{eq:model_rb}). However, the simulation of the rHeston model (\ref{eq:model_rh}) is significantly slower, so we still need to rely on the dual mesh approach for its simulation.

We again simulate $N=100000$ trajectories, and we use $n=800$ for the rBergomi (\ref{eq:model_rb}), Heston (\ref{eq:model_h}) and Bergomi (\ref{eq:model_b}) models. For the dual mesh for the rHeston model (\ref{eq:model_rh}), we use a fine mesh with $n=800$ and a coarse mesh with $n=300$. The optimised time-independent parameters are shown in Table \ref{tab:res_param2}, and in Figure \ref{fig:res_param2} we can observe the time-dependent parameters of each model.
\begin{table}[!htbp]
\centering
\begin{tabular}{lccccc}
\toprule
Volatility model   & $H$    & $\eta$  & $\rho$  & $\kappa$ & $\hat{v}_0$ \\
\midrule
rBergomi         & 0.0460 & 2.5600  & -0.2110 & --       & --          \\
rHeston          & 0.3214 & 2.7121  & -0.1677 & 4.5931   & --          \\
1-factor Bergomi & --     & 14.6416 & -0.1481 & 36.8716  & --          \\
Heston           & --     & 13.3771 & -0.1429 & 5.3333   & 0.1194      \\
\bottomrule
\end{tabular}
\caption{Optimised time-independent parameters of the models calibrated to the data from 4th June 2025.}\label{tab:res_param2}
\end{table}

\begin{figure}[!hp]%
	\centering
	\subfloat[rBergomi: $\xi_0(t)$]{
		{\includegraphics[width=0.45\textwidth]{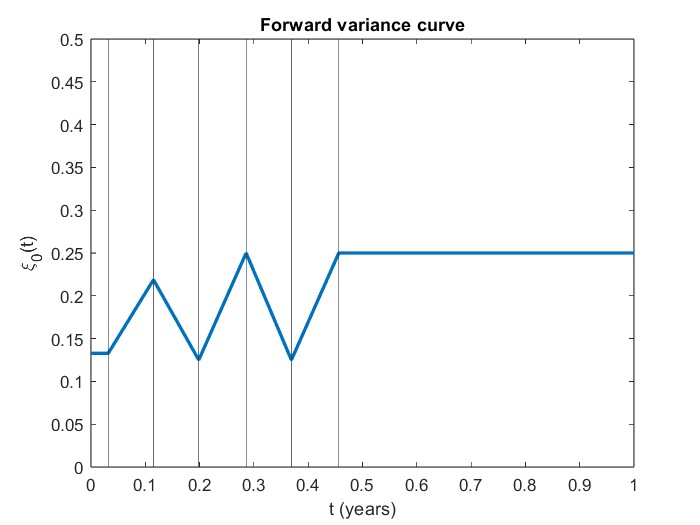} }
	}%
	\subfloat[rHeston: $\xi_0(t)$]{
		{\includegraphics[width=0.45\textwidth]{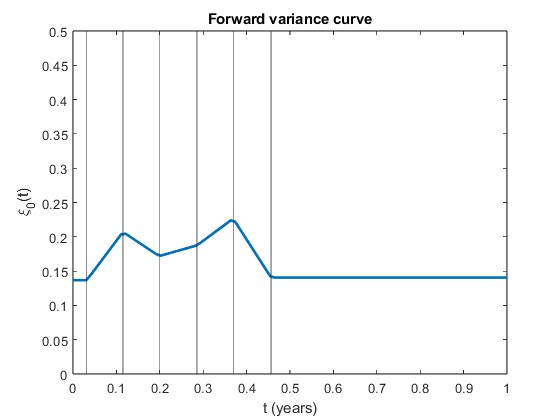} }
	}%
	
	\subfloat[Bergomi: $\xi_0(t)$]{
		{\includegraphics[width=0.45\textwidth]{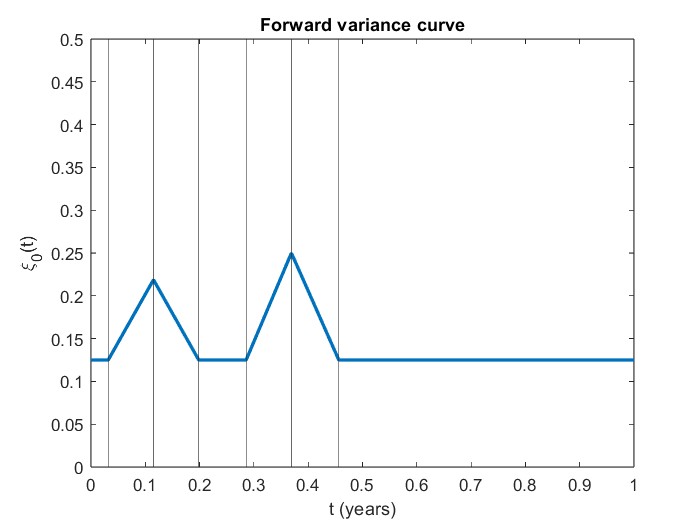} }
	}%
	\subfloat[Heston: $\bar{v}(t)$]{
		{\includegraphics[width=0.45\textwidth]{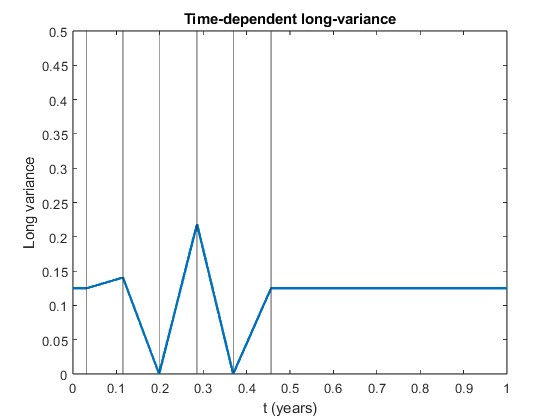} }
	}%
	
	\caption{Optimised time-dependent parameter of each model calibrated to the data from 4th June 2025. The vertical lines mark the maturity dates considered.}%
	\label{fig:res_param2}%
\end{figure}

Figures \ref{fig:smile_rb2} to \ref{fig:smile_h2} show the volatility smiles produced by the respective models. These results are consistent with those obtained in Section \ref{sec:dataset1}: the rough models obtained slightly better results than their classical counterparts, but the latter are able to reproduce the market data at the cost of very large vol-vol $\eta$ and mean reversion speed $\kappa$. Furthermore, the rBergomi model outperforms the rHeston model, particularly for the closest maturity, and again the Hurst parameters $H$ obtained from each rough model differ considerably.

We note that every model presents difficulties in capturing the right-hand wing of the smile with the closest maturity, even after introducing the penalty term in the loss function. Interestingly, the rHeston model performs the worst in this region, while the rBergomi model provides the closest approximation to the quoted implied volatility.

\begin{figure}[!hp]%
	\centering
	\subfloat[Loss $= 0.1002$]{
		{\includegraphics[width=0.45\textwidth]{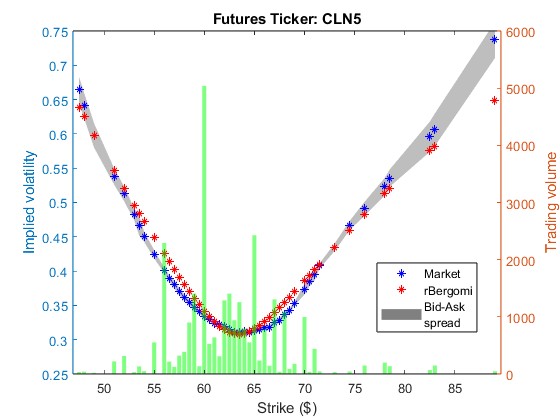} }
	}%
	\subfloat[Loss $= 0.0133$]{
		{\includegraphics[width=0.45\textwidth]{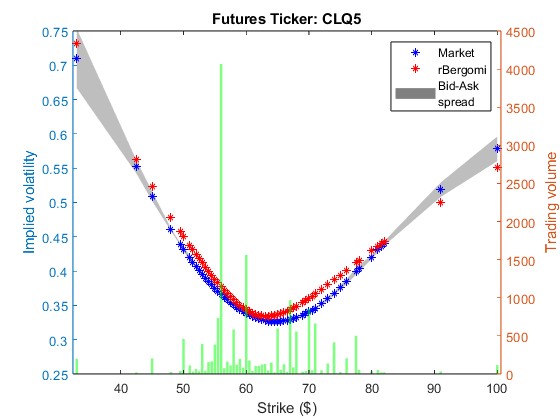} }
	}%
	
	\subfloat[Loss $= 0.0052$]{
		{\includegraphics[width=0.45\textwidth]{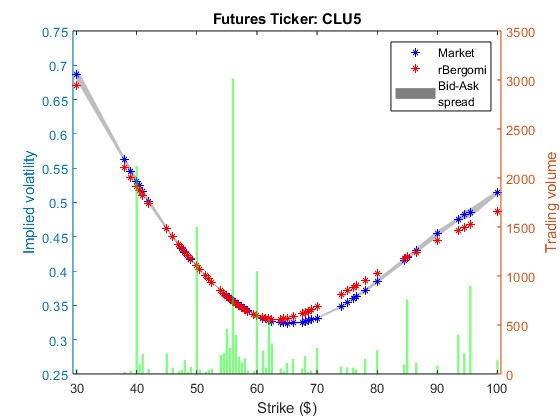} }
	}%
	\subfloat[Loss $= 0.0047$]{
		{\includegraphics[width=0.45\textwidth]{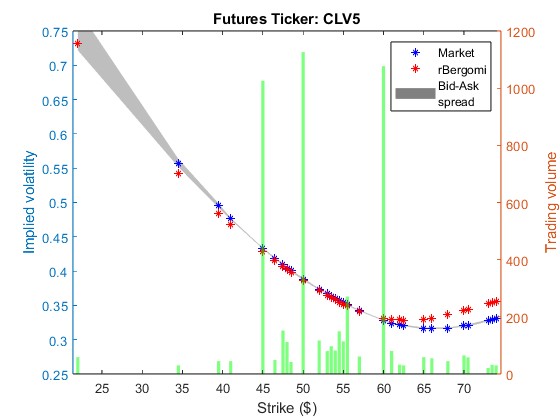} }
	}%
	
	\subfloat[Loss $= 0.0099$]{
		{\includegraphics[width=0.45\textwidth]{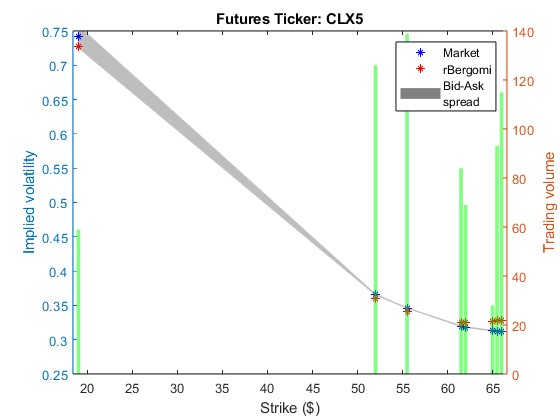} }
	}%
	\subfloat[Loss $= 0.0842$]{
		{\includegraphics[width=0.45\textwidth]{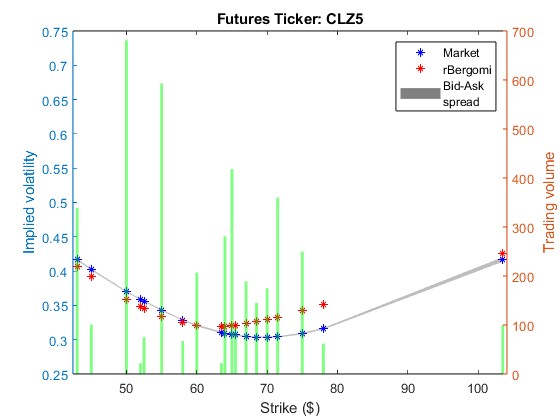} }
	}%
	\caption{Calibration results of the rBergomi model (\ref{eq:model_rb}) for the data from 4th June 2025.}%
	\label{fig:smile_rb2}%
\end{figure}


\begin{figure}[!hp]%
	\centering
	\subfloat[Loss $= 0.7082$]{
		{\includegraphics[width=0.45\textwidth]{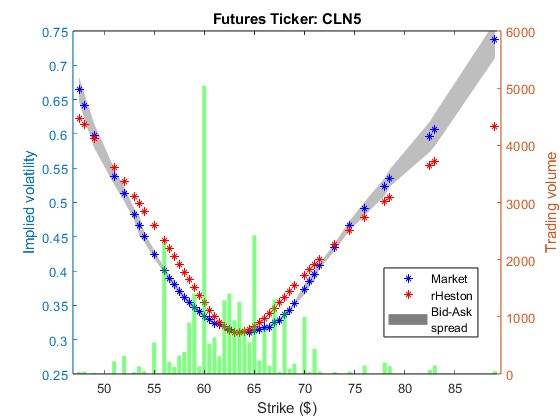} }
	}%
	\subfloat[Loss $= 0.0571$]{
		{\includegraphics[width=0.45\textwidth]{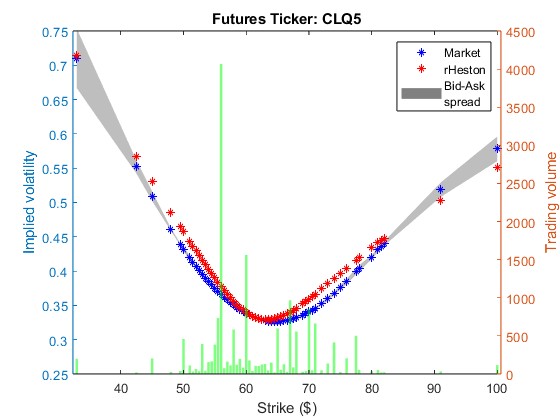} }
	}%
	
	\subfloat[Loss $= 0.0114$]{
		{\includegraphics[width=0.45\textwidth]{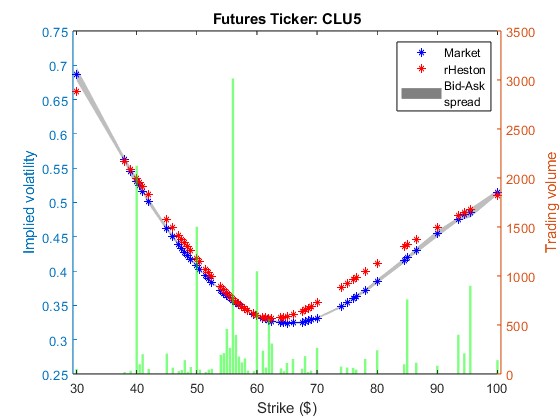} }
	}%
	\subfloat[Loss $= 0.0508$]{
		{\includegraphics[width=0.45\textwidth]{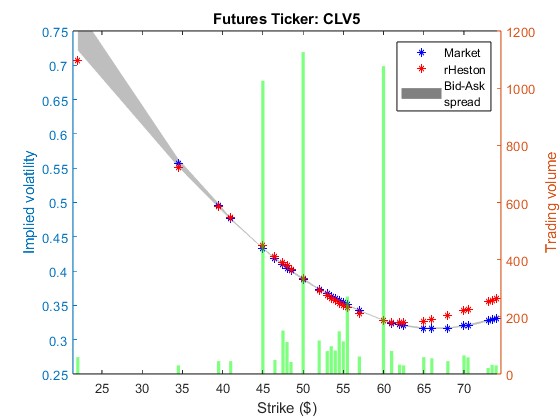} }
	}%
	
	\subfloat[Loss $= 0.0508$]{
		{\includegraphics[width=0.45\textwidth]{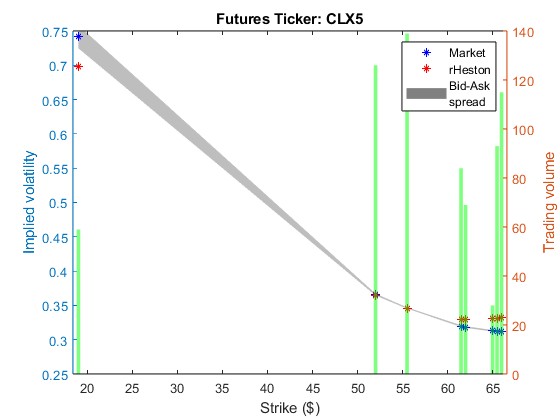} }
	}%
	\subfloat[Loss $= 0.0940$]{
		{\includegraphics[width=0.45\textwidth]{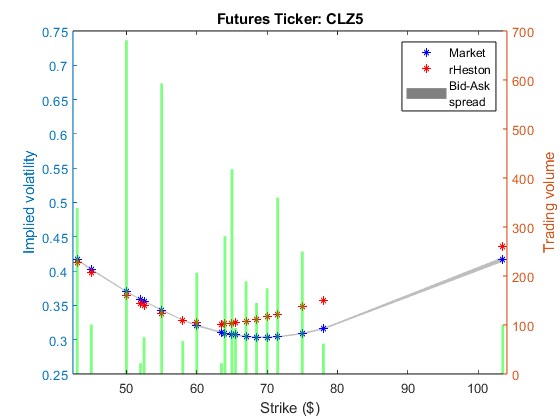} }
	}%
	\caption{Calibration results of the rHeston model (\ref{eq:model_rh}) for the data from 4th June 2025.}%
	\label{fig:smile_rh2}%
\end{figure}


\begin{figure}[!hp]%
	\centering
	\subfloat[Loss $= 0.2445$]{
		{\includegraphics[width=0.45\textwidth]{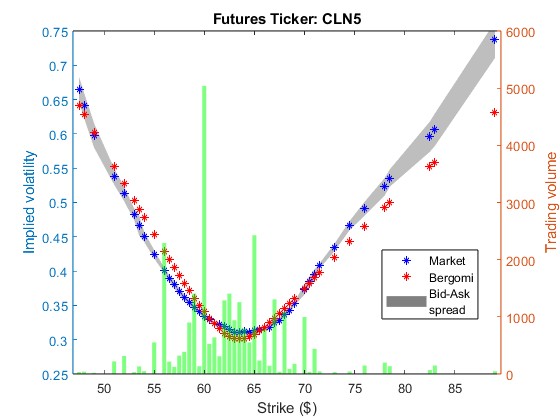} }
	}%
	\subfloat[Loss $= 0.0489$]{
		{\includegraphics[width=0.45\textwidth]{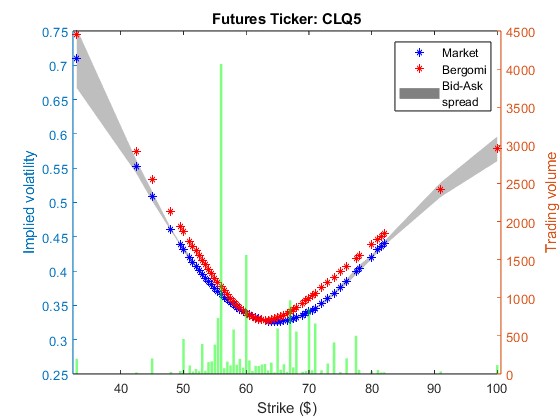} }
	}%
	
	\subfloat[Loss $= 0.0396$]{
		{\includegraphics[width=0.45\textwidth]{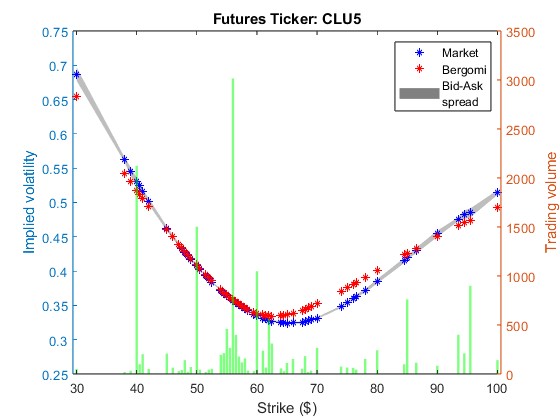} }
	}%
	\subfloat[Loss $= 0.2361$]{
		{\includegraphics[width=0.45\textwidth]{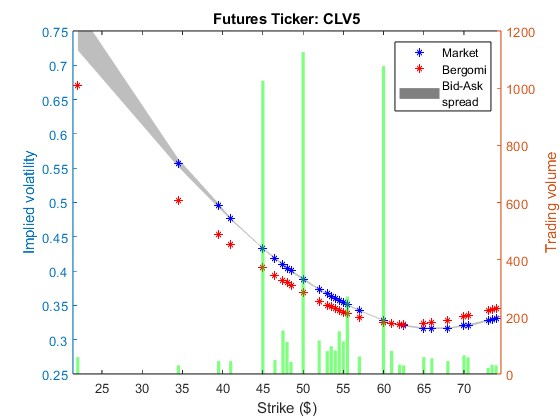} }
	}%
	
	\subfloat[Loss $= 0.1181$]{
		{\includegraphics[width=0.45\textwidth]{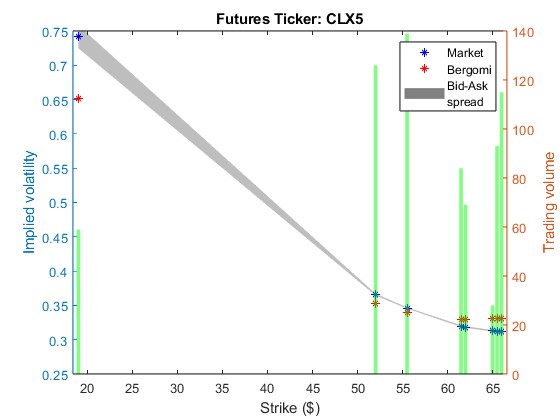} }
	}%
	\subfloat[Loss $= 0.1924$]{
		{\includegraphics[width=0.45\textwidth]{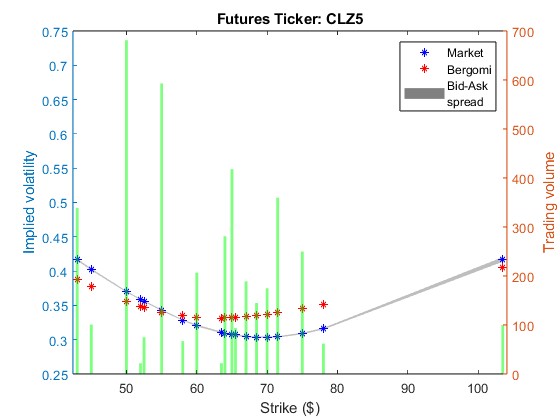} }
	}%
	\caption{Calibration results of the Bergomi model (\ref{eq:model_b}) for the data from 4th June 2025.}%
	\label{fig:smile_b2}%
\end{figure}


\begin{figure}[!hp]%
	\centering
	\subfloat[Loss $= 0.5915$]{
		{\includegraphics[width=0.45\textwidth]{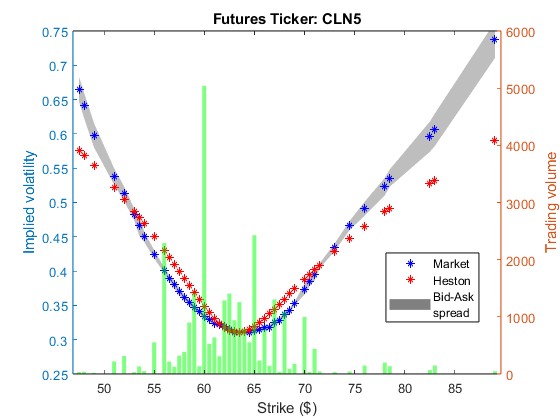} }
	}%
	\subfloat[Loss $= 0.0148$]{
		{\includegraphics[width=0.45\textwidth]{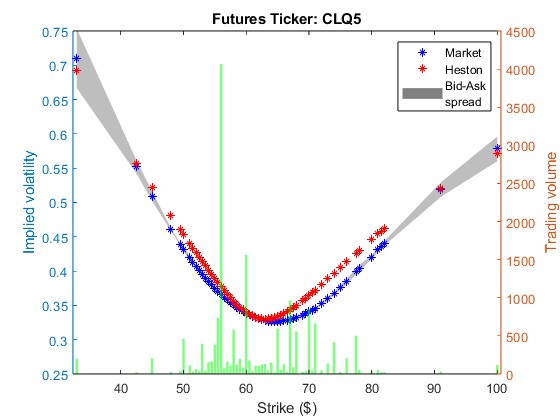} }
	}%
	
	\subfloat[Loss $= 0.0470$]{
		{\includegraphics[width=0.45\textwidth]{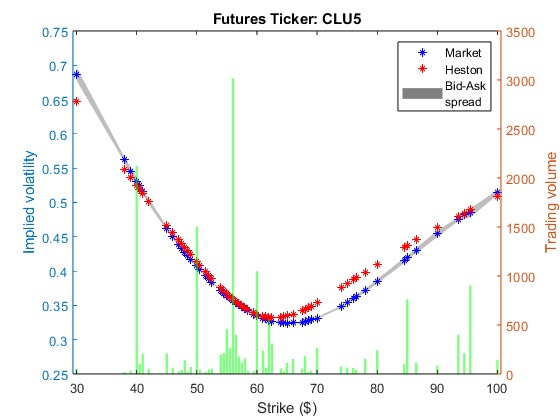} }
	}%
	\subfloat[Loss $= 0.0788$]{
		{\includegraphics[width=0.45\textwidth]{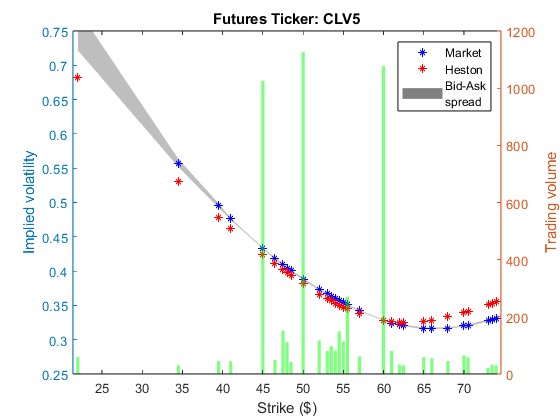} }
	}%
	
	\subfloat[Loss $= 0.1115$]{
		{\includegraphics[width=0.45\textwidth]{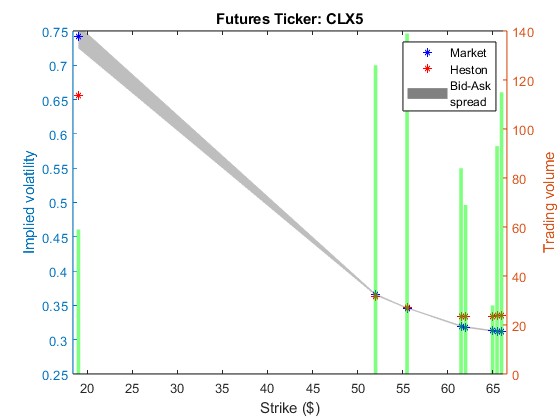} }
	}%
	\subfloat[Loss $= 0.0853$]{
		{\includegraphics[width=0.45\textwidth]{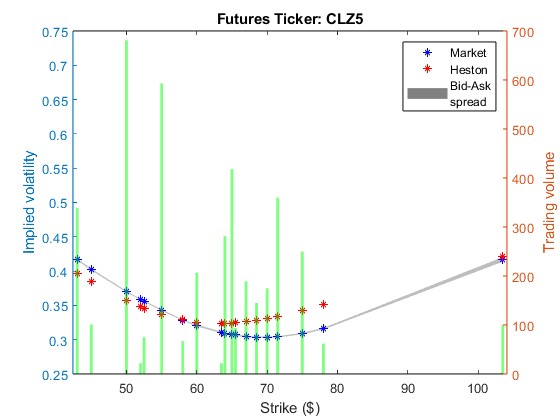} }
	}%
	\caption{Calibration results of the Heston model (\ref{eq:model_h}) for the data from 4th June 2025.}%
	\label{fig:smile_h2}%
\end{figure}

\FloatBarrier

\subsection{Time-dependent correlation}
In order to improve the previous results with minimal changes to the model, we allow the correlation parameter to become time-dependent. This is the only parameter of our model that affects just the price process and not the volatility, apart from the mean reversion speed, which we do not calibrate. Consequently, this extension requires minimal changes to the implementation. Specifically, we consider a different correlation for each maturity, yielding a piecewise-constant structure.

This extension is also reasonable from a theoretical perspective. Indeed, allowing the correlation to be time-dependent is equivalent to considering a multi-factor model, in which there is one spot price process for each single maturity. The main drawback of this approach is that these price processes can no longer be interpreted as real prices. However, this is not problematic in our framework, since we are already working with the fictitious spot price process.

\begin{table}[!ht]
\centering
\begin{tabular}{clccc}
\toprule
Data-set                         & Volatility model & $H$    & $\eta$ & $\kappa$ \\
\midrule
\multirow{2}{*}{14th March 2025} & rBergomi         & 0.0603 & 2.189  & --       \\
                                 & rHeston          & 0.315  & 2.7643 & 8.3101   \\
\multirow{2}{*}{4th June 2025}   & rBergomi         & 0.0453 & 2.5596 & --       \\
                                 & rHeston          & 0.2759 & 2.4718 & 4.5016   \\
\bottomrule
\end{tabular}
\caption{Optimised time-independent parameters of the extended rough models with time-dependent correlation.}\label{tab:corrtimedep}
\end{table}

We use the same simulation parameters $n$ and $N$ as in the previous sections. The calibrated model parameters are shown in Table \ref{tab:corrtimedep}. For clarity, the calibrated correlation is not included in the table, and it is instead displayed in Figures \ref{fig:res_corr1} and \ref{fig:res_corr2}, jointly with the calibrated initial forward variance curve. The obtained parameters are consistent with those reported in Tables \ref{tab:res_param1} and \ref{tab:res_param2}, in the case of the rBergomi model. Furthermore, we observe that the correlation decreases for longer maturities, with the previously calibrated constant correlation lying within the range of the time-dependent curve.

\begin{figure}[ht!]%
	\centering
	\subfloat[rBergomi: $\rho(t)$]{
		{\includegraphics[width=0.35\textwidth]{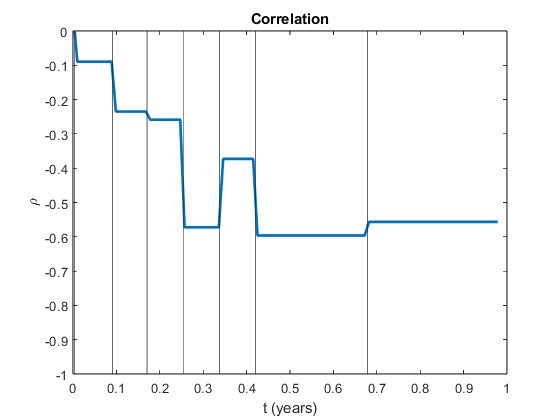} }
	}%
    \subfloat[rBergomi: $\xi_0(t)$]{
		{\includegraphics[width=0.35\textwidth]{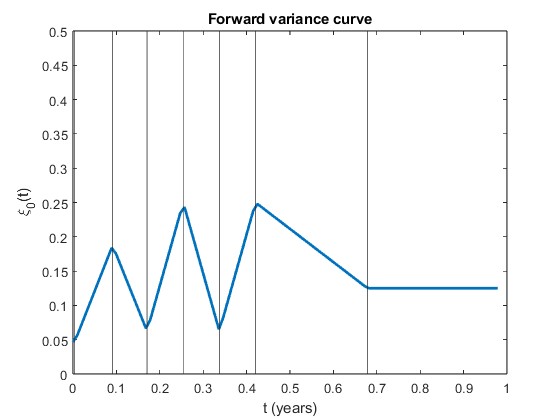} }
	}%

    \subfloat[rHeston: $\rho(t)$]{
		{\includegraphics[width=0.35\textwidth]{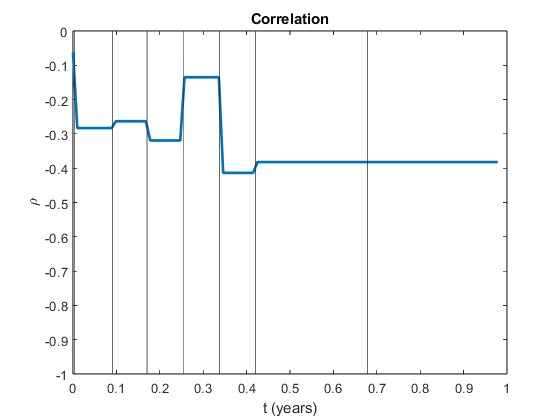} }
	}%
	\subfloat[rHeston: $\xi_0(t)$]{
		{\includegraphics[width=0.35\textwidth]{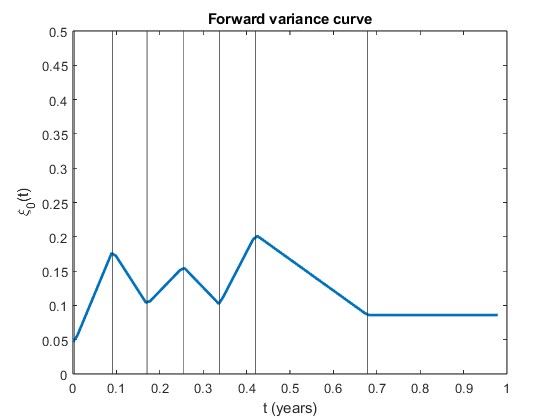} }
	}%
	
	\caption{Time-dependent correlation and initial forward variance curve calibrated to the data from 14th March 2025. The vertical lines mark the maturity dates considered.}%
	\label{fig:res_corr1}%
\end{figure}

\begin{figure}[ht!]%
	\centering
	\subfloat[rBergomi: $\rho(t)$]{
		{\includegraphics[width=0.35\textwidth]{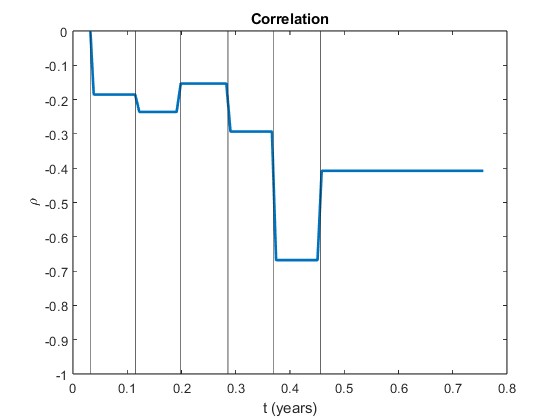} }
	}%
    \subfloat[rBergomi: $\xi_0(t)$]{
		{\includegraphics[width=0.35\textwidth]{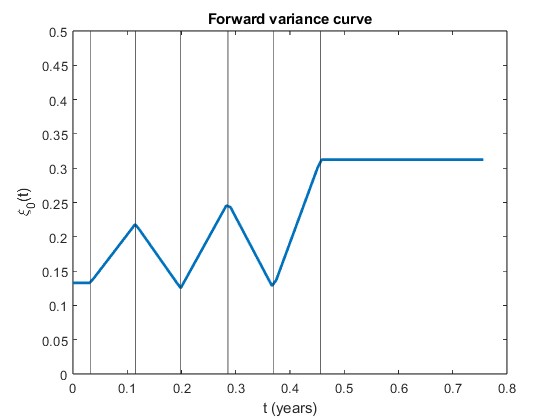} }
	}%

    \subfloat[rHeston: $\rho(t)$]{
		{\includegraphics[width=0.35\textwidth]{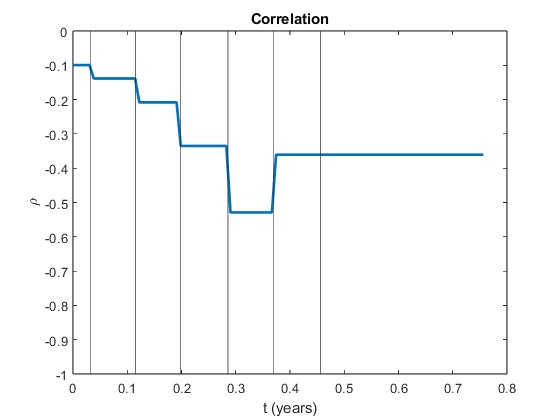} }
	}%
	\subfloat[rHeston: $\xi_0(t)$]{
		{\includegraphics[width=0.35\textwidth]{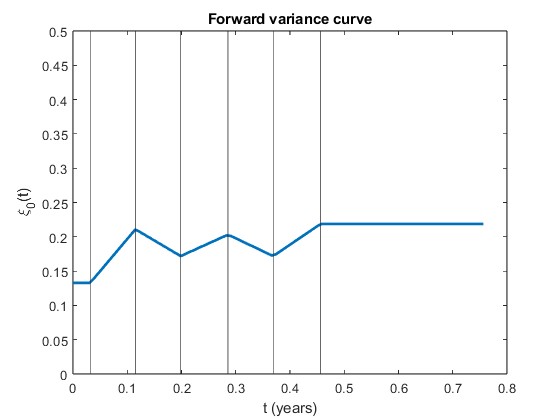} }
	}%
	
	\caption{Time-dependent correlation and initial forward variance curve calibrated to the data from 4th June 2025. The vertical lines mark the maturity dates considered.}%
	\label{fig:res_corr2}%
\end{figure}

The resulting smiles can be found in Figures \ref{fig:smile_rb1_corrtimedep} to \ref{fig:smile_rh2_corrtimedep}. We observe that allowing the correlation to depend on the maturity date improved the fitting of the volatility surface, especially in the case of the rHeston model. Indeed, on both data sets, there is a significant improvement for the shorter and longer maturity dates. Very similar improvements, although less pronounced, are found in the rBergomi model, particularly in the right hand side of smiles with longer maturity dates.

\begin{figure}[ht!]%
	\centering
	\subfloat[Loss $= 0.0098$]{
		{\includegraphics[width=0.45\textwidth]{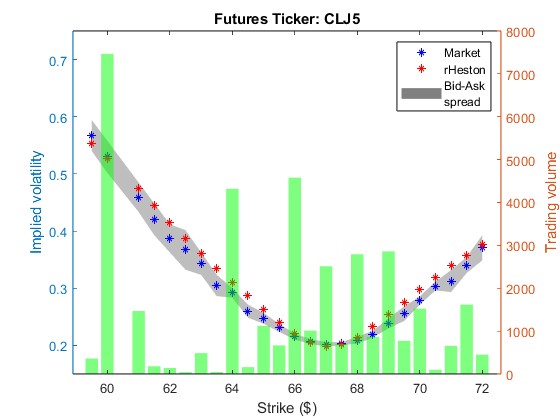} }
	}%
	\subfloat[Loss $= 0.0397$]{
		{\includegraphics[width=0.45\textwidth]{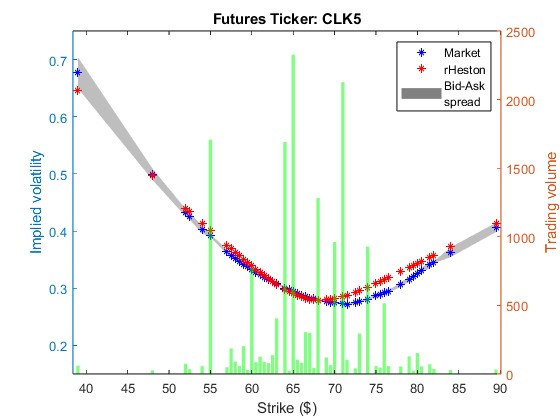} }
	}%
	
	\subfloat[Loss $= 0.0078$]{
		{\includegraphics[width=0.45\textwidth]{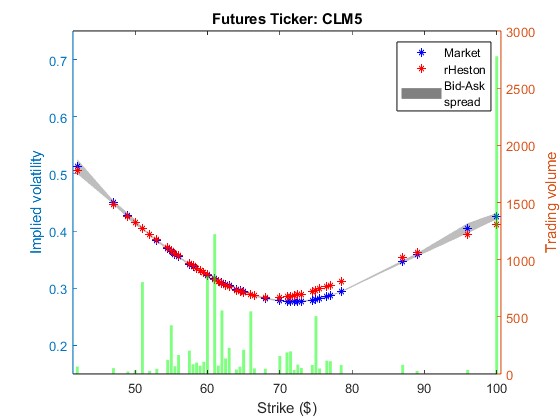} }
	}%
	\subfloat[Loss $= 0.0118$]{
		{\includegraphics[width=0.45\textwidth]{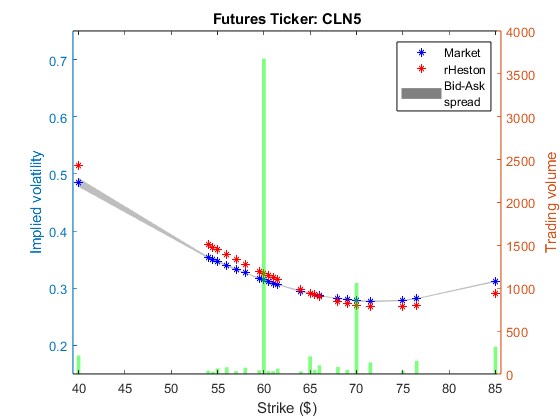} }
	}%
	
	\subfloat[Loss $= 0.0085$]{
		{\includegraphics[width=0.45\textwidth]{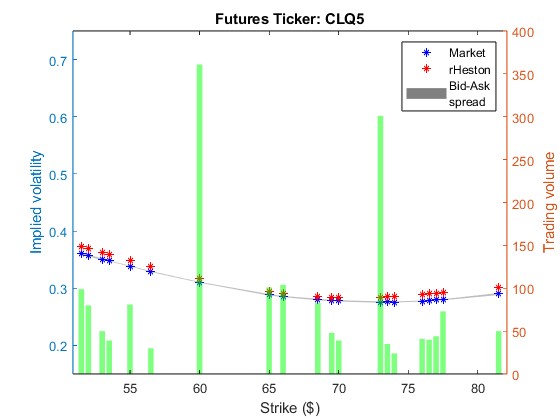} }
	}%
	\subfloat[Loss $= 0.0102$]{
		{\includegraphics[width=0.45\textwidth]{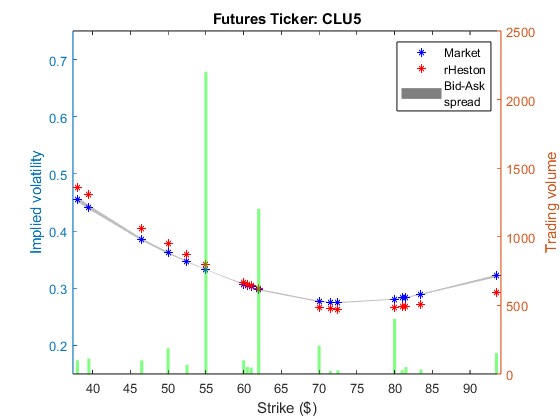} }
	}%

    \subfloat[Loss $= 0.0054$]{
		{\includegraphics[width=0.45\textwidth]{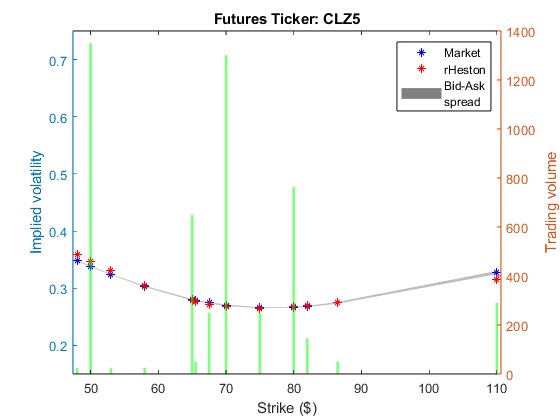} }
	}%
	\caption{Calibration results of the rBergomi model (\ref{eq:model_rb}) with time-dependent correlation for the data from 14th March 2025.}%
	\label{fig:smile_rb1_corrtimedep}%
\end{figure}


\begin{figure}[ht!]%
	\centering
	\subfloat[Loss $= 0.1741$]{
		{\includegraphics[width=0.45\textwidth]{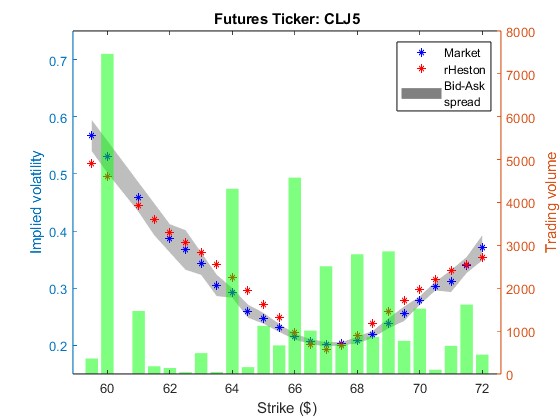} }
	}%
	\subfloat[Loss $= 0.0103$]{
		{\includegraphics[width=0.45\textwidth]{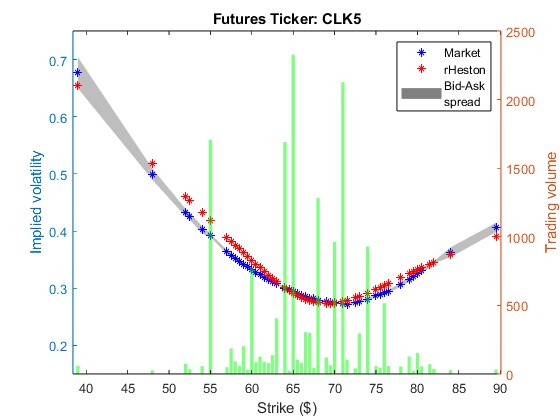} }
	}%
	
	\subfloat[Loss $= 0.0130$]{
		{\includegraphics[width=0.45\textwidth]{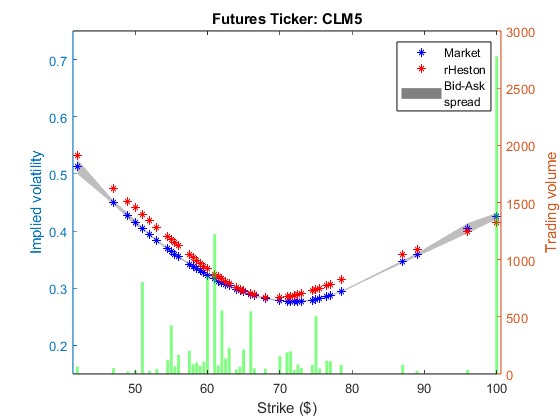} }
	}%
	\subfloat[Loss $= 0.0039$]{
		{\includegraphics[width=0.45\textwidth]{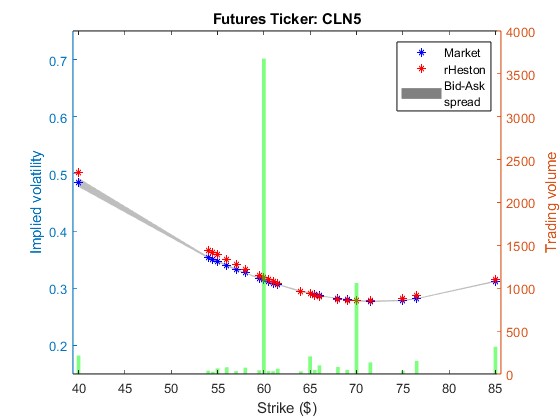} }
	}%
	
	\subfloat[Loss $= 0.0120$]{
		{\includegraphics[width=0.45\textwidth]{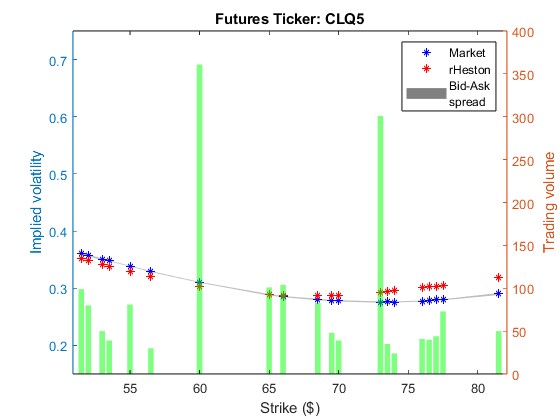} }
	}%
	\subfloat[Loss $= 0.0038$]{
		{\includegraphics[width=0.45\textwidth]{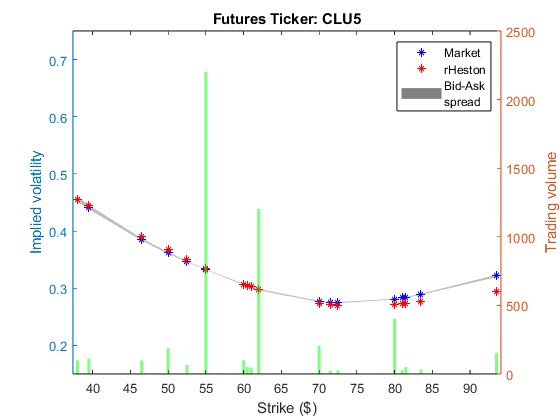} }
	}%

    \subfloat[Loss $= 0.0033$]{
		{\includegraphics[width=0.45\textwidth]{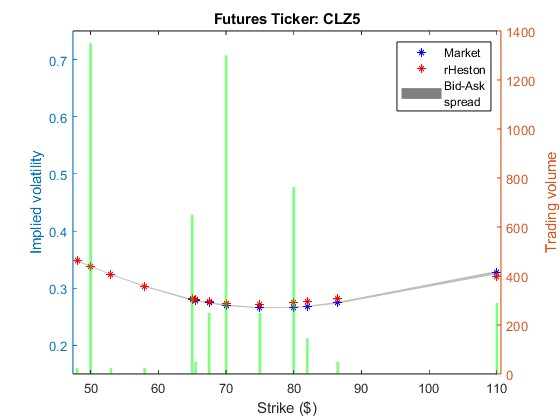} }
	}%
	\caption{Calibration results of the rHeston model (\ref{eq:model_rh}) with time-dependent correlation for the data from 14th March 2025.}%
	\label{fig:smile_rh1_corrtimedep}%
\end{figure}


\begin{figure}[ht!]%
	\centering
	\subfloat[Loss $= 0.0961$]{
		{\includegraphics[width=0.45\textwidth]{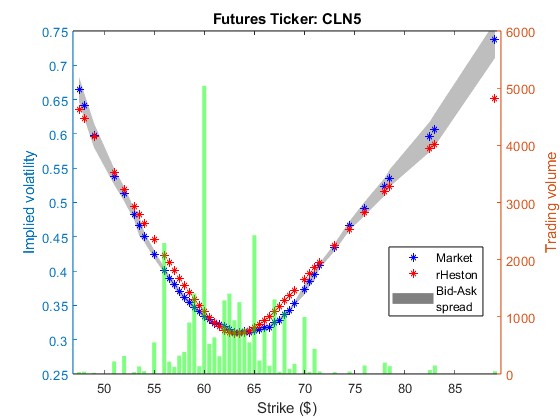} }
	}%
	\subfloat[Loss $= 0.0442$]{
		{\includegraphics[width=0.45\textwidth]{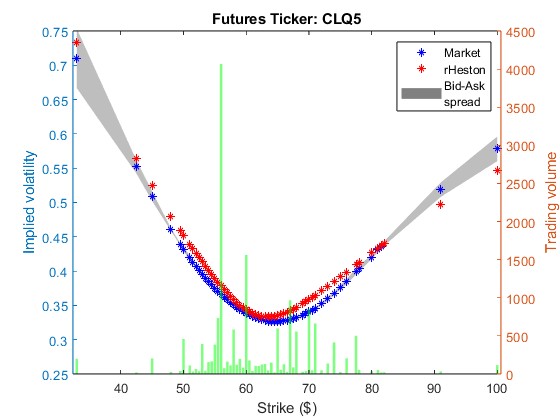} }
	}%
	
	\subfloat[Loss $= 0.0073$]{
		{\includegraphics[width=0.45\textwidth]{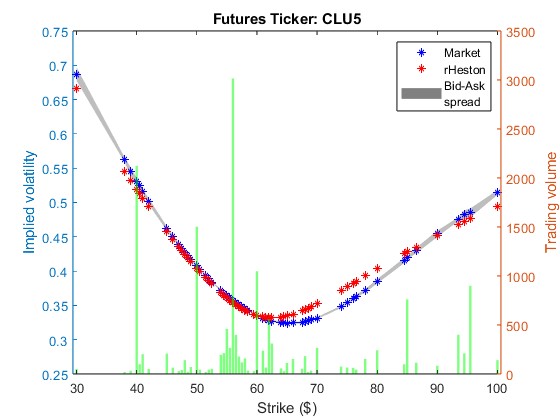} }
	}%
	\subfloat[Loss $= 0.0025$]{
		{\includegraphics[width=0.45\textwidth]{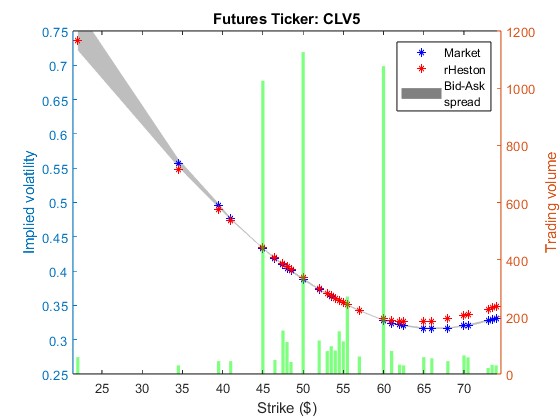} }
	}%
	
	\subfloat[Loss $= 0.0110$]{
		{\includegraphics[width=0.45\textwidth]{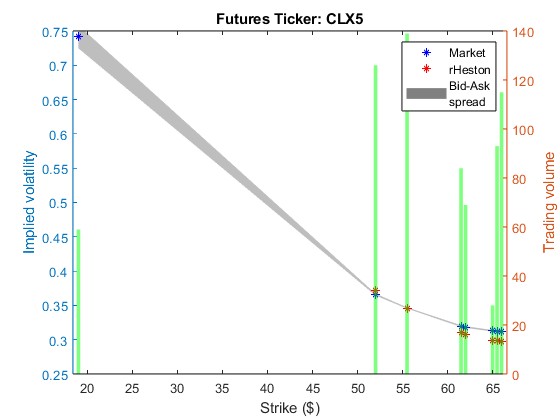} }
	}%
	\subfloat[Loss $= 0.0072$]{
		{\includegraphics[width=0.45\textwidth]{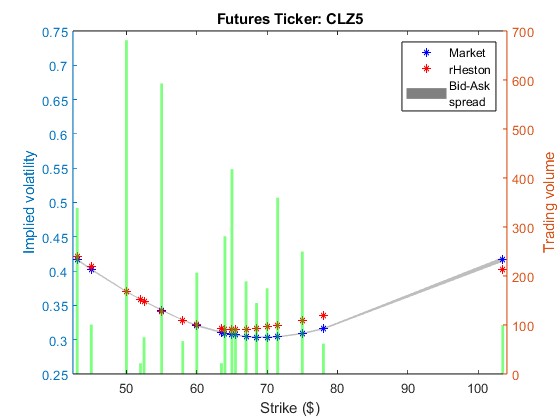} }
	}%
	\caption{Calibration results of the rBergomi model (\ref{eq:model_rb}) with time-dependent correlation for the data from 4th June 2025.}%
	\label{fig:smile_rb2_corrtimedep}%
\end{figure}


\begin{figure}[ht!]%
	\centering
		\subfloat[Loss $= 0.3198$]{
		{\includegraphics[width=0.45\textwidth]{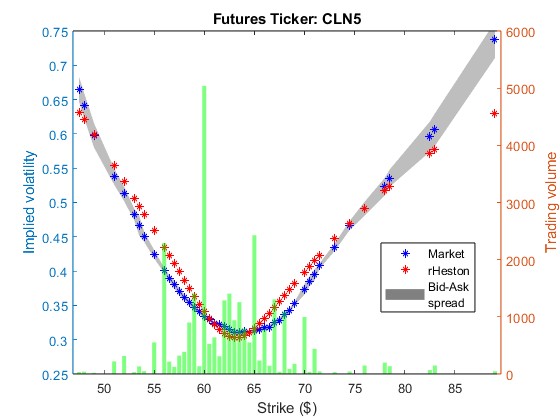} }
	}%
	\subfloat[Loss $= 0.0138$]{
		{\includegraphics[width=0.45\textwidth]{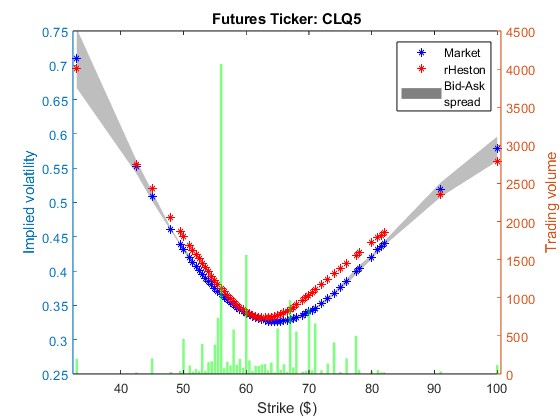} }
	}%
	
	\subfloat[Loss $= 0.0082$]{
		{\includegraphics[width=0.45\textwidth]{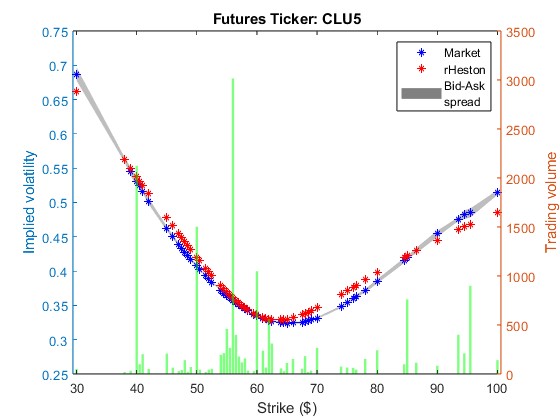} }
	}%
	\subfloat[Loss $= 0.0132$]{
		{\includegraphics[width=0.45\textwidth]{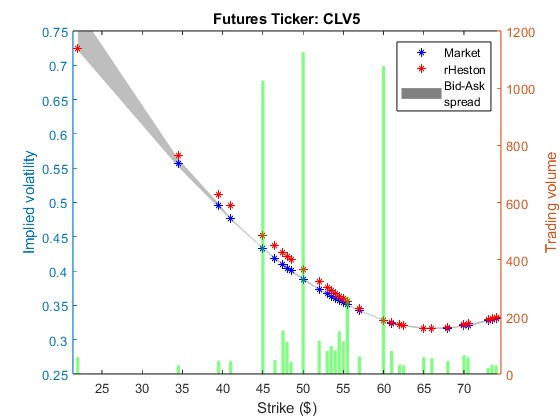} }
	}%
	
	\subfloat[Loss $= 0.0114$]{
		{\includegraphics[width=0.45\textwidth]{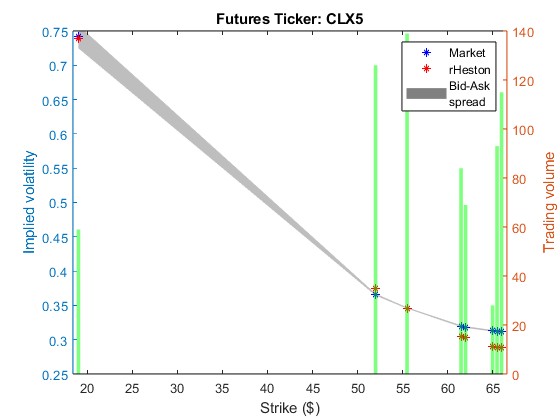} }
	}%
	\subfloat[Loss $= 0.0043$]{
		{\includegraphics[width=0.45\textwidth]{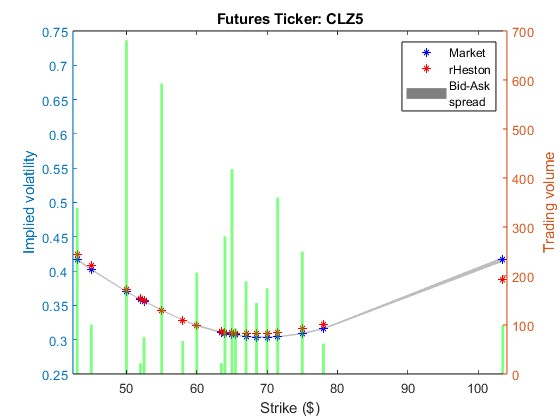} }
	}%
	\caption{Calibration results of the rHeston model (\ref{eq:model_rh}) with time-dependent correlation for the data from 4th June 2025.}%
	\label{fig:smile_rh2_corrtimedep}%
\end{figure}

\FloatBarrier

\section{Conclusions}\label{sec:conclusion}
Prior to this work, little to no information could be found in the literature on the applicability of rough volatility models in commodity markets. Our contribution consists of a novel model to price futures options on commodities, in which the volatility process is given by a general forward volatility model, and in particular by a general rough volatility model satisfying the condition that the process $Y_t$, defined in (\ref{eq:y_t}), is a martingale. We provided a theoretical justification for this model and calibrated it to market data on WTI Crude Oil.

For the considered data sets, we found that the rBergomi and rHeston models yielded good results, but the classical Bergomi and Heston models were able to nearly reproduce their results at the cost of very large vol-vol $\eta$ and mean reversion speed $\kappa$. Further numerical tests with different data sets are necessary to explain this phenomenon. Furthermore, we found that allowing the correlation to depend on maturity improves the fit of the volatility surface, especially for the rHeston model.

Several directions for future research remain open. In addition to investigating the qualitative properties of the Bergomi and Heston models with very large $\eta$ and $\kappa$, it would also be interesting to apply our model and methodology to other commodities besides WTI Crude Oil. The results in \citet{alfeuspaper} show that many commodities exhibit a Hurst parameter $H<0.2$; thus, in principle, model (\ref{eq:model}) should yield good results in a wide variety of commodity markets.

\section*{Acknowledgements}

H. Folgar-Came\'an and C. V\'azquez acknowledge the funding from Ministry of Science and Innovation of Spain through the grant PID2022-141058OB-I00, as well as from the Galician Government through the grant ED431C 2022/047 (both including FEDER financial support). Both authors also acknowledge the support of CITIC, as a center accredited for excellence within the Galician University System and a
member of the CIGUS Network, that receives subsidies from the Department of Education, Science, Universities, and Vocational Training of the Xunta de Galicia. Additionally, CITIC is co-financed by the EU through the FEDER Galicia 2021-27 operational program (Ref. ED431G 2023/01).

\section*{Disclaimer}

The authors report no potential competing interests. The opinions expressed in this document are solely those of the authors and do not represent in any way those of their present and past employers.

\bibliographystyle{plainnat}
\bibliography{arxiv_refs}

@ARTICLE{alfeuspaper,
  author  = {Alfeus, M. and Nikitopoulos, C.},
  title   = {Forecasting volatility in commodity markets with long-memory models},
  journal = {Journal of Commodity Markets},
  volume  = {28},
  year    = {2022},
  pages   = {100248},
  doi     = {https://doi.org/10.1016/j.jcomm.2022.100248}
}

@ARTICLE{pakkanenpaper,
  author  = {Bennedsen, M. and Lunde, A. and Pakkanen, M. S.},
  title   = {Hybrid scheme for {B}rownian semistationary processes},
  journal = {Finance and Stochastics},
  volume  = {21},
  year    = {2017},
  pages   = {931--965},
  doi     = {https://doi.org/10.1007/s00780-017-0335-5}
}

@ARTICLE{baschettipaper,
  author  = {Baschetti, F. and Bormetti, G. and Rossi, P.},
  title   = {Deep calibration with random grids},
  journal = {Quantitative Finance},
  volume  = {24},
  number  = {9},
  year    = {2024},
  pages   = {1263--1285},
  doi     = {https://doi.org/10.1080/14697688.2024.2332375}
}

@BOOK{bergomibook,
  author    = {Bergomi, L.},
  title     = {Stochastic Volatility Modeling},
  edition   = {1},
  publisher = {CRC Press},
  year      = {2015}
}

@ARTICLE{bessembinderpaper,
  author  = {Bessembinder, H. and Coughenour, J. F. and Seguin, P. J. and Smoller, M. M.},
  title   = {Is there a term structure of futures volatilities? {R}eevaluating the {S}amuelson hypothesis},
  journal = {Journal of Derivatives},
  volume  = {4},
  number  = {2},
  year    = {1995},
  pages   = {45--58},
  doi     = {https://doi.org/10.3905/jod.1996.407967}
}

@ARTICLE{blackscholespaper,
  author  = {Black, F. and Scholes, M.},
  title   = {The pricing of options and corporate liabilities},
  journal = {Journal of Political Economy},
  volume  = {81},
  number  = {3},
  year    = {1973},
  pages   = {637--654},
  doi     = {https://doi.org/10.1086/260062}
}

@UNPUBLISHED{bourgeypaper,
  author = {Bourgey, F. and Delemotte, J. and De Marco, S.},
  title  = {Smile dynamics and rough volatility},
  note   = {Available at SSRN: \url{https://ssrn.com/abstract=4911186}},
  year   = {2024}
}

@ARTICLE{buehlerpaper,
  author  = {Buehler, H.},
  title   = {Consistent variance curve models},
  journal = {Finance and Stochastics},
  volume  = {10},
  number  = {2},
  year    = {2006},
  pages   = {178--203},
  doi     = {https://doi.org/10.1007/s00780-006-0008-2}
}

@ARTICLE{drimuspaper,
  author  = {Drimus, G. and Farkas, W.},
  title   = {Local volatility of volatility for the {VIX} market},
  journal = {Review of Derivatives Research},
  volume  = {16},
  number  = {3},
  year    = {2013},
  pages   = {267--293},
  doi     = {https://doi.org/10.1007/s11147-012-9086-9}
}

@ARTICLE{eleuchpaper,
  author  = {El Euch, B. O. and Rosenbaum, M.},
  title   = {Perfect hedging in rough {H}eston models},
  journal = {The Annals of Applied Probability},
  volume  = {28},
  number  = {6},
  year    = {2018},
  pages   = {3813--3856},
  doi     = {https://doi.org/10.1214/18-AAP1408}
}

@ARTICLE{fukasawapaper,
  author  = {Fukasawa, M.},
  title   = {Short-time at-the-money skew and rough fractional volatility},
  journal = {Quantitative Finance},
  volume  = {17},
  number  = {2},
  year    = {2017},
  pages   = {189--198},
  doi     = {https://doi.org/10.1080/14697688.2016.1197410}
}

@ARTICLE{gassiatpaper,
  author  = {Gassiat, P.},
  title   = {On the martingale property in the rough {B}ergomi model},
  journal = {Electronic Communications in Probability},
  volume  = {24},
  year    = {2019},
  pages   = {1--9},
  doi     = {https://doi.org/10.1214/19-ECP239}
}

@ARTICLE{hqepaper,
  author  = {Gatheral, J.},
  title   = {Efficient simulation of affine forward variance models},
  journal = {Risk Magazine},
  year    = {2022},
  note    = {February 2022, available in \url{https://www.risk.net/cutting-edge/banking/7923331/efficient-simulation-of-affine-forward-variance-models}}
}

@ARTICLE{gatheralpaper,
  author  = {Gatheral, J. and Jaisson, T. and Rosenbaum, M.},
  title   = {Volatility is rough},
  journal = {Quantitative Finance},
  volume  = {18},
  number  = {6},
  year    = {2014},
  pages   = {933--949},
  doi     = {https://doi.org/10.1080/14697688.2017.1393551}
}

@ARTICLE{hopaper,
  author  = {Ho, C. C. and Lee, P. H. and Tsai, P. S.},
  title   = {Competing hypotheses on the {S}amuelson effect in futures markets},
  journal = {Applied Economics},
  volume  = {55},
  number  = {20},
  year    = {2023},
  pages   = {2261--2272},
  doi     = {https://doi.org/10.1080/00036846.2022.2102569}
}

@BOOK{ikedabook,
  author    = {Ikeda, N. and Watanabe, S.},
  title     = {Stochastic Differential Equations and Diffusion Processes},
  edition   = {2},
  series    = {North-Holland Mathematical Library},
  publisher = {Elsevier Science Publishers B.V.},
  address   = {Amsterdam},
  year      = {1989}
}

@ARTICLE{kolmogorovpaper,
  author  = {Kolmogorov, A.},
  title   = {Wiener spirals and some other interesting curves in {H}ilbert space},
  journal = {Dokl. Akad. Nauk SSSR},
  volume  = {26},
  number  = {2},
  year    = {1940},
  pages   = {115--118}
}

@ARTICLE{mandelbrotpaper,
  author  = {Mandelbrot, B. and Van Ness, J.},
  title   = {Fractional {B}rownian motions, fractional noises and applications},
  journal = {SIAM Review},
  volume  = {10},
  number  = {4},
  year    = {1968},
  pages   = {422--437},
  doi     = {https://doi.org/10.1137/1010093}
}

@BOOK{mikoschbook,
  author    = {Mikosch, T.},
  title     = {Elementary Stochastic Calculus with Finance in View},
  edition   = {1},
  series    = {Advanced Series on Statistical Science \& Applied Probability},
  publisher = {World Scientific Publishing},
  address   = {Singapore},
  year      = {2004}
}

@BOOK{mishurabook,
  author    = {Mishura, Y. S.},
  title     = {Stochastic Calculus for Fractional {B}rownian Motion and Related Processes},
  edition   = {1},
  series    = {Lecture Notes in Mathematics},
  publisher = {Springer},
  address   = {Berlin, Heidelberg},
  year      = {2008}
}

@ARTICLE{nunnomishurapaper,
  author  = {Di Nunno, G. and Kubilius, K. and Mishura, Y. and Yurchenko-Tytarenko, A.},
  title   = {From constant to rough: A survey of continuous volatility modeling},
  journal = {Mathematics},
  volume  = {11},
  number  = {19},
  year    = {2023},
  pages   = {4201},
  doi     = {https://doi.org/10.3390/math11194201}
}

@ARTICLE{pallavicinipaper,
  author  = {Nastasi, E. and Pallavicini, A. and Sartorelli, G.},
  title   = {Smile modelling in commodity markets},
  journal = {International Journal of Theoretical and Applied Finance},
  volume  = {23},
  number  = {3},
  year    = {2020},
  pages   = {1--28},
  doi     = {https://doi.org/10.1142/S0219024920500193}
}

@ARTICLE{rogerspaper,
  author  = {Rogers, L. C. G.},
  title   = {Arbitrage with fractional {B}rownian motion},
  journal = {Mathematical Finance},
  volume  = {7},
  number  = {1},
  year    = {1997},
  pages   = {95--105},
  doi     = {https://doi.org/10.1111/1467-9965.00025}
}

@ARTICLE{samuelsonpaper,
  author  = {Samuelson, P. A.},
  title   = {Proof that properly anticipated prices fluctuate randomly},
  journal = {Industrial Management Review},
  volume  = {6},
  number  = {2},
  year    = {1965},
  pages   = {41--49}
}

@ARTICLE{swishchukpaper,
  author  = {Swishchuk, A.},
  title   = {Explicit option pricing formula for mean-reverting asset in energy market},
  journal = {Journal of Numerical and Applied Mathematics},
  volume  = {96},
  year    = {2008},
  pages   = {216--233}
}

@ARTICLE{takaishipaper,
  author  = {Takaishi, T.},
  title   = {Rough volatility of {B}itcoin},
  journal = {Financial Research Letters},
  volume  = {32},
  year    = {2020},
  pages   = {101379},
  doi     = {https://doi.org/10.1016/j.frl.2019.101379}
}

@article{bayer,
author = {Christian Bayer and Peter Friz and Jim Gatheral},
title = {Pricing under rough volatility},
journal = {Quantitative Finance},
volume = {16},
number = {6},
pages = {887--904},
year = {2016},
publisher = {Routledge},
doi = {https://doi.org/10.1080/14697688.2015.1099717}
}

@article{abijaber,
author = {Abi Jaber, Eduardo and El Euch, Omar},
title = {Multifactor Approximation of Rough Volatility Models},
journal = {SIAM Journal on Financial Mathematics},
volume = {10},
number = {2},
pages = {309-349},
year = {2019},
doi = {https://doi.org/10.1137/18M1170236}
}

\end{document}